\documentclass[twocolumn,showpacs,showkeys,prx]{revtex4}%
\usepackage{todonotes}
\usepackage{times}
\usepackage{amsfonts}
\usepackage{amsmath}
\usepackage{amssymb}
\usepackage{graphicx,color}%
\usepackage{graphics}
\usepackage{txfonts}
\usepackage{listings}
\usepackage{algorithm}
\usepackage{algpseudocode}
\providecommand{\U}[1]{\protect\rule{.1in}{.1in}}
\newtheorem{theorem}{Theorem}

\newtheorem{lemma}[theorem]{Lemma}

\newtheorem{proposition}[theorem]{Proposition}
\newtheorem{remark}[theorem]{Remark}

\newenvironment{proof}[1][Proof]{\noindent\textbf{#1.} }{\ \hfill\rule{0.5em}{0.5em}}

\def\bA{\textbf A}
\def\bb{\textbf b}
\def\cM{\cal M}

\begin{document}
\title{Spanning connectivity in a multilayer network and its relationship to site-bond percolation}

\author{Saikat Guha$^1$, Donald Towsley$^2$, Philippe Nain$^3$, {\c C}a{\u g}atay {\c C}apar$^4$, Ananthram Swami$^5$, Prithwish Basu$^1$}
\affiliation{
$^1$\textit{Raytheon BBN Technologies, Cambridge, MA 02138, USA},\\
$^2$\textit{University of Massachusetts, Amherst, MA 01003, USA,}\\
$^3$\textit{Inria, 06902 Sophia Antipolis Cedex, France}\\
$^4$\textit{Ericsson Research, San Jose, CA 95134, USA,}\\
$^5$\textit{US Army Research Laboratory, Adelphi, MD 20783, USA}}

\begin{abstract}
We analyze the connectivity of an $M$-layer network over a common set of nodes that are active only in a fraction of the layers. Each layer is assumed to be a subgraph (of an underlying connectivity graph $G$) induced by each node being active in any given layer with probability $q$. The $M$-layer network is formed by aggregating the edges over all $M$ layers. We show that when $q$ exceeds a threshold $q_c(M)$, a giant connected component appears in the $M$-layer network---thereby enabling far-away users to connect using `bridge' nodes that are active in multiple network layers---even though the individual layers may only have small disconnected islands of connectivity. We show that $q_c(M) \lesssim \sqrt{-\ln(1-p_c)}\,/{\sqrt{M}}$, where $p_c$ is the bond percolation threshold of $G$, and $q_c(1) \equiv q_c$ is its site percolation threshold. We find $q_c(M)$ exactly for when $G$ is a large random network with an arbitrary node-degree distribution. We find $q_c(M)$ numerically for various regular lattices, and find an exact lower bound for the kagome lattice. Finally, we find an intriguingly close connection between this multilayer percolation model and the well-studied problem of site-bond percolation, in the sense that both models provide a smooth transition between the traditional site and bond percolation models. Using this connection, we translate known analytical approximations of the site-bond critical region, which are functions only of $p_c$ and $q_c$ of the respective lattice, to excellent general approximations of the multilayer connectivity threshold $q_c(M)$.
\end{abstract}
\maketitle

The last few years has seen a surge of interest in multilayer networks, several properties of various genres of which have been studied, much of which has been covered in these two review articles~\cite{Kiv14, Boc14}. Specific example studies include the diffusion dynamics of multilayer networks~\cite{Gom13}, cascades~\cite{Bru12, Gao13}, spectral properties~\cite{Rib13}, robustness analysis stemming from overlapping multilayer links~\cite{Cel13}, growing random multilayer networks~\cite{Nic13}, epidemic spread~\cite{Mar11}, a tensorial formulation~\cite{Dom13}, and algorithmic complexity of finding short paths through co-evolving multilayer networks~\cite{Bas15}. The connectivity properties of random multilayer networks have also been studied, such as the study of the properties of the giant connected component (GCC) in a random network with correlated multiplexicity, i.e., where the node degree distributions across layers have positive (or negative) correlations~\cite{Lee12}. 

The multilayer network model we study in this paper was inspired by a multi-channel wireless adhoc communication network~\cite{WNAN}, where each node only uses a small subset of all the available channels at any given time (to save energy---battery life of a radio transceiver for instance), and the consideration of the minimum number of channels in which each node should be active to ensure long range connectivity.

\begin{figure}
\centering
\includegraphics[width=0.9\columnwidth]{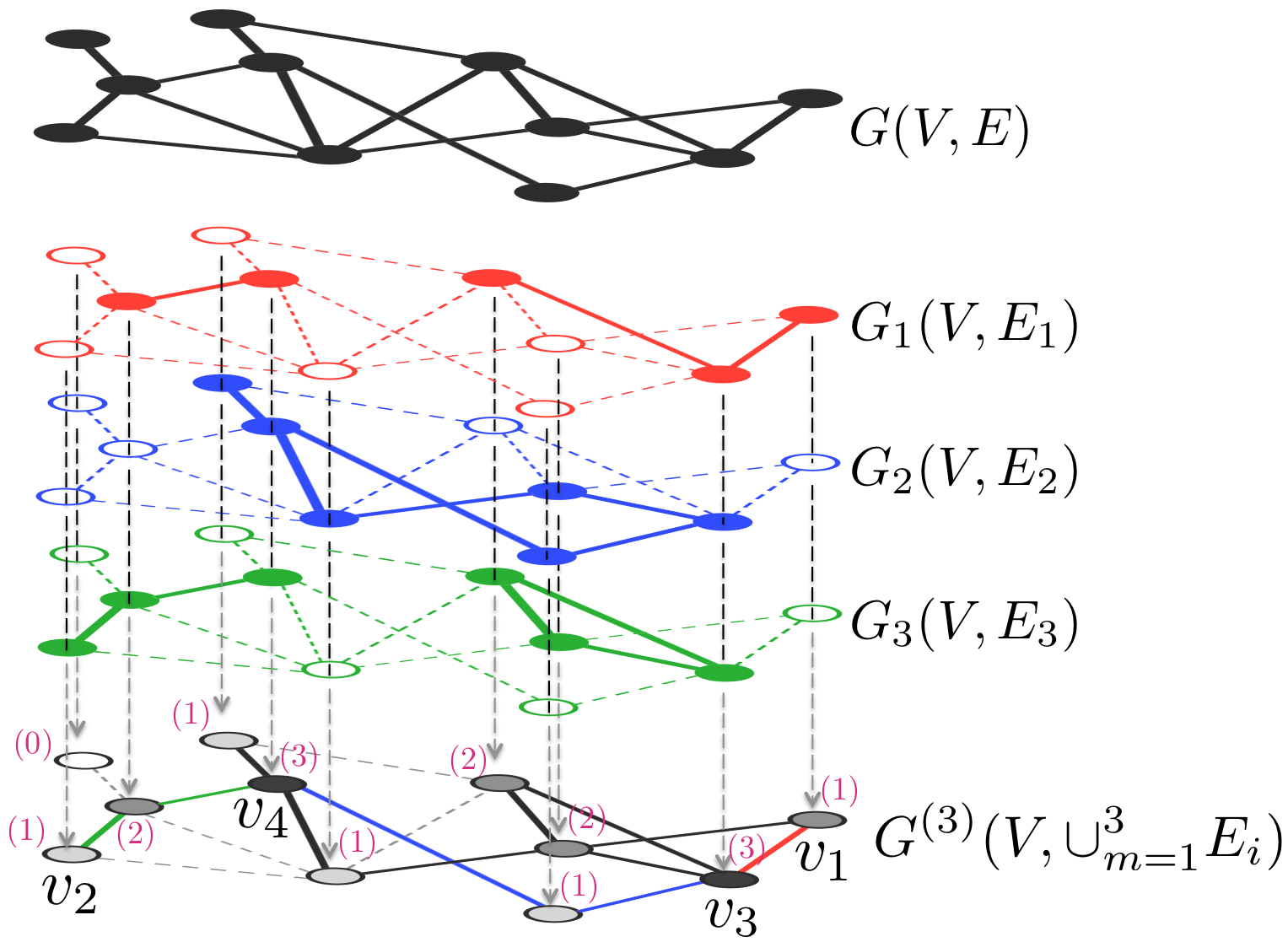}
\caption{(Color online) Schematic of a $3$-layer network. The numbers of layers in which each node is occupied (active) are shown.}
\label{fig:setup}
\end{figure}
We consider a set of users connected via $M$ co-existing networks $G_1, \ldots, G_M$. Let us assume that each user (node) is active only in a subset of these networks. Consequently, a user who is active in both $G_1$ and $G_2$ can help connect two other users that are active in $G_1$ alone, and in $G_2$ alone, respectively, by forming a bridge. Fig.~\ref{fig:setup} illustrates an example with $M=3$ networks (`layers'), where a path connecting $v_1$ and $v_2$ must traverse all three layers, and one such path is shown to go through the bridge nodes $v_3$ and $v_4$, both of which are occupied in more than one layer. 

Some concrete examples of such multilayer networks are: (1) a network of cities connected via different airline companies where each city is served only by a subset of all the airlines~\cite{Bas15, Buc13}, (2) a network of users with accounts on multiple online social networks~\cite{Mur14}, and (3) a military communication network of units equipped with radios that can listen and transmit simultaneously on a subset of multiple frequencies~\cite{WNAN}. Each of these scenarios have one feature in common: the multilayer network is formed over a {\em common} set of nodes via co-existing means of connectivity. In other words, each node in the multilayer network is one single entity (e.g., a city, a social network user, or a multi-channel radio) that may be active simultaneously in a subset of multiple layers, where each layer that a given node is active in, provides a distinct mode for that node to connect to its neighboring nodes that are also active in that layer. 

In our analysis in this paper, we will make a simplifying assumption, that each network layer is a subgraph of a {\em common} underlying connectivity graph $G(V, E)$ whose edge set $E$ defines all the {\em possible} connections, some of which may be dormant if the two nodes an edge connects are not active in at least one common layer. The underlying connectivity graphs for the aforesaid examples are: the network of airway passages connecting the cities, the underlying friendship network (who is a friend of whom on social networks), and the Euclidean geometric graph induced by the locations of the multichannel radios and their maximum range, respectively.

Each node will be assumed to be active in a given layer with probability $q$, and the node occupancies in each layer will be taken to be independent. The subgraph corresponding to the $m$-th layer $G_m(V, E_m)$ is obtained by removing all the edges of $G(V, E)$ both of whose end nodes are not active in the $m$-th layer. The merged (random) graph $G^{(M)}(V, \cup_{m=1}^M E_m) \subset G$ represents the effective multilayer network whose edges can connect the nodes in $V$. So, the $M$-layer graph $G^{(M)}$ is formed by directly aggregating the edges over all $M$ layers. For an edge to exist in $G^{(M)}$, the two nodes it connects must be active in at least one common layer. Since each node is a single physical entity, one can think of the inter-layer connectivity graph at a given node to be a $k$-clique, where $k$ is the number of layers that node is occupied in. Our goal in this paper is to study the threshold value of the single-layer node-occupation probability $q$, which we will denote $q_c(M)$, when a GCC (or a spanning cluster) appears in the $M$-layer network $G^{(M)}$, at which point distant users can connect using a series of bridge nodes that are active in multiple layers, even though the individual layers may only have small disconnected islands of connectivity. Clearly, for $M=1$, this model reduces to the standard i.i.d. site percolation problem, and thus $q_c(1) = q_c$, the site-percolation threshold of $G$. In Fig.~\ref{fig:squaregrid_picture}, we show an illustrative numerical example for $2$ layers over a square grid.

As the reader may already have noted, the model we analyze in this paper is insufficient to accurately model most practical multilayer networks. For instance, the assumption that each individual network layer is an induced subgraph of a common underlying graph may not be accurate. For example, in a multilayer social network, a node's neighbors (friends) in Linkedin may be different from its neighbors in Facebook. On the other hand, in the multichannel wireless ad hoc network example described above, the assumption that each layer samples from one underlying connectivity graph is quite accurate. Several other interesting extensions of our model are described in the Conclusions section of this paper.


\begin{figure}[htb]
\centering
\includegraphics[width=\columnwidth]{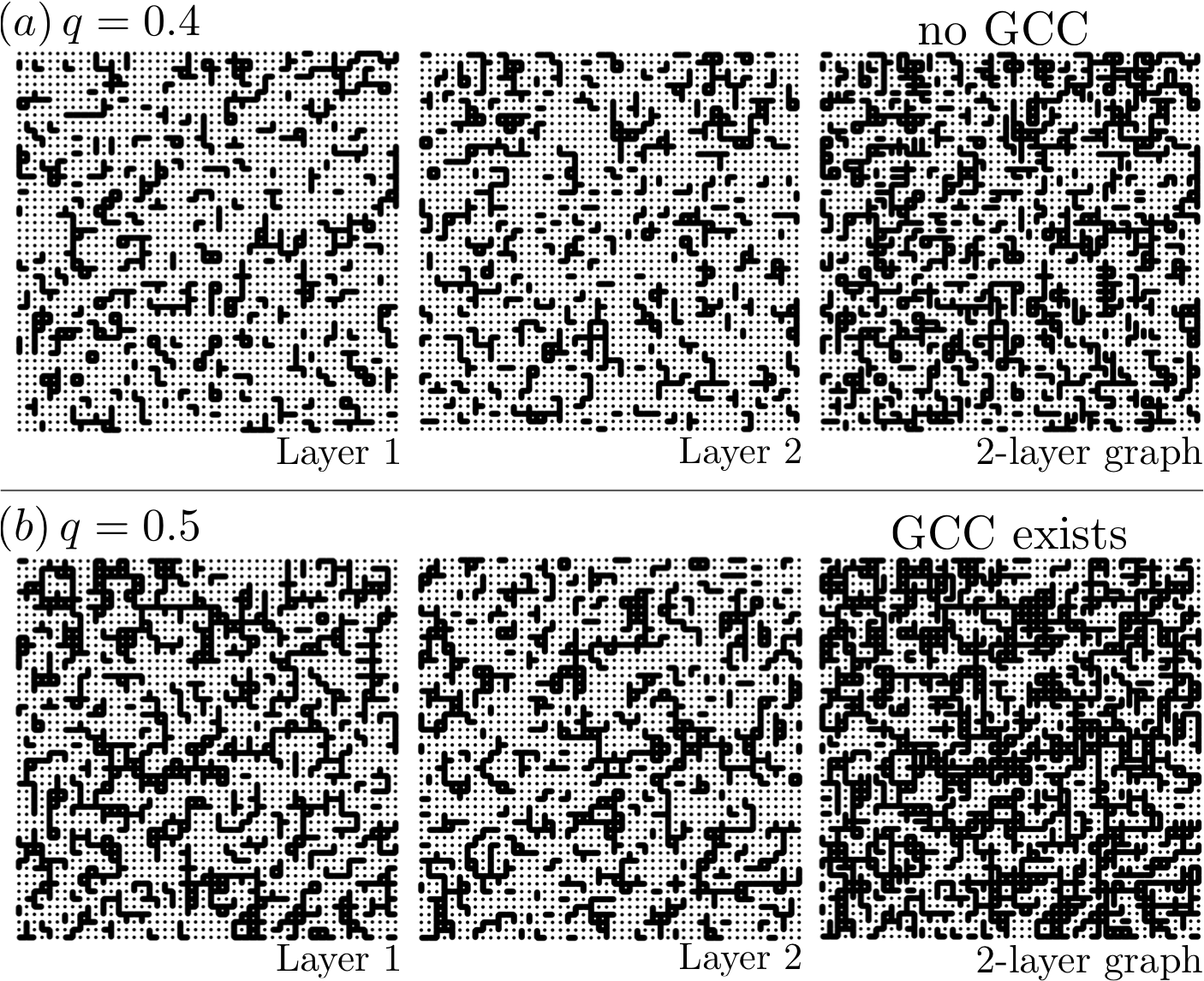}
\caption{Two independent site-percolation instances (`layers') of a square grid $G$ are shown, for site-occupation probability (a) $q = 0.4$ and (b) $q = 0.5$. Each site in each layer is occupied with probability $q$, and a bond is activated when both sites at its end points are occupied. Active bonds are shown by black line segments. A $2$-layer lattice is formed by first marking as occupied all those sites in $G$ that are occupied in at least one of two independent {\em layers} (site-percolation instances of $G$), and then activating all the bonds that have sites at both their end points occupied. The site-percolation threshold of a square grid, $q_c \equiv q_c(1) \approx 0.59$, which is why the single layer graphs do not have a giant connected component (GCC) either for $q=0.4$ or $q=0.5$. However, the $2$-layer percolation threshold for the square lattice, $q_c(2) \approx 0.47$. Thus, the $2$-layer graph created with $q=0.4$ does not exhibit a GCC, whereas the one created with $q=0.5$ does.}
\label{fig:squaregrid_picture}
\end{figure}
The first major contribution of this paper, described in Section~\ref{sec:exact}, is the detailed study of the behavior of $q_c(M)$ for a general underlying connectivity graph $G$. We show that $q_c(M) \sim 1/\sqrt{M}$ for $M$ large. This implies that when each node is occupied in roughly $c\sqrt{M}$ (of the $M$) layers, $c$ being a constant, spanning connectivity emerges in the $M$-layer network. We show that $c$ approaches $\sqrt{-\ln(1-p_c)}$ as $M \to \infty$, where $p_c$ is the bond-percolation threshold of $G$. In Section~\ref{sec:random}, we find $q_c(M)$ exactly when $G$ is a large random network with an arbitrary node-degree distribution. In Section~\ref{sec:lattice}, we evaluate $q_c(M)$ numerically for various regular lattices, and find an analytical lower bound for $q_c(M)$ for when $G$ is a regular kagome lattice. We show, for a general graph $G$, that $q_c(M) \lesssim \sqrt{-\ln(1-p_c)}\,/{\sqrt{M}}$, where $p_c$ is the bond percolation threshold of $G$, and that the inequality is asymptotically tight when $M \to \infty$. Clearly, for $M=1$, $q_c(1) \equiv q_c$, where $q_c$ is the site percolation threshold of $G$. Therefore as $M$ goes from $1$ to $\infty$, $q_c(M)$ goes from being a function solely of $q_c$ to being solely a function of $p_c$. This suggests that our multilayer percolation model forms a smooth transition between the standard site and bond percolation models. This leads to our second main contribution, described in Section~\ref{sec:sitebond}: an intriguingly close connection between the aforesaid multilayer percolation model and the well-studied problem of site-bond (or, mixed) percolation---a percolation process defined on the single-layer graph $G$, in which each site and each bond in $G$ is independently activated with probability $q$ and $p$, respectively. Both models provide a smooth transition between the traditional i.i.d. site and bond percolation models. Using this connection, we show a way to translate analytical approximations to the site-bond critical region (the region in the $(p, q)$ space where a GCC exists with high probability) that are functions solely of $p_c$ and $q_c$, to an excellent general approximation of the multilayer percolation thresholds, $q_c(M)$. We conclude the paper in Section~\ref{sec:conclusions}.

\section{Multilayer percolation in a large graph}\label{sec:exact}

\begin{figure}
\centering
\includegraphics[width=0.9\columnwidth]{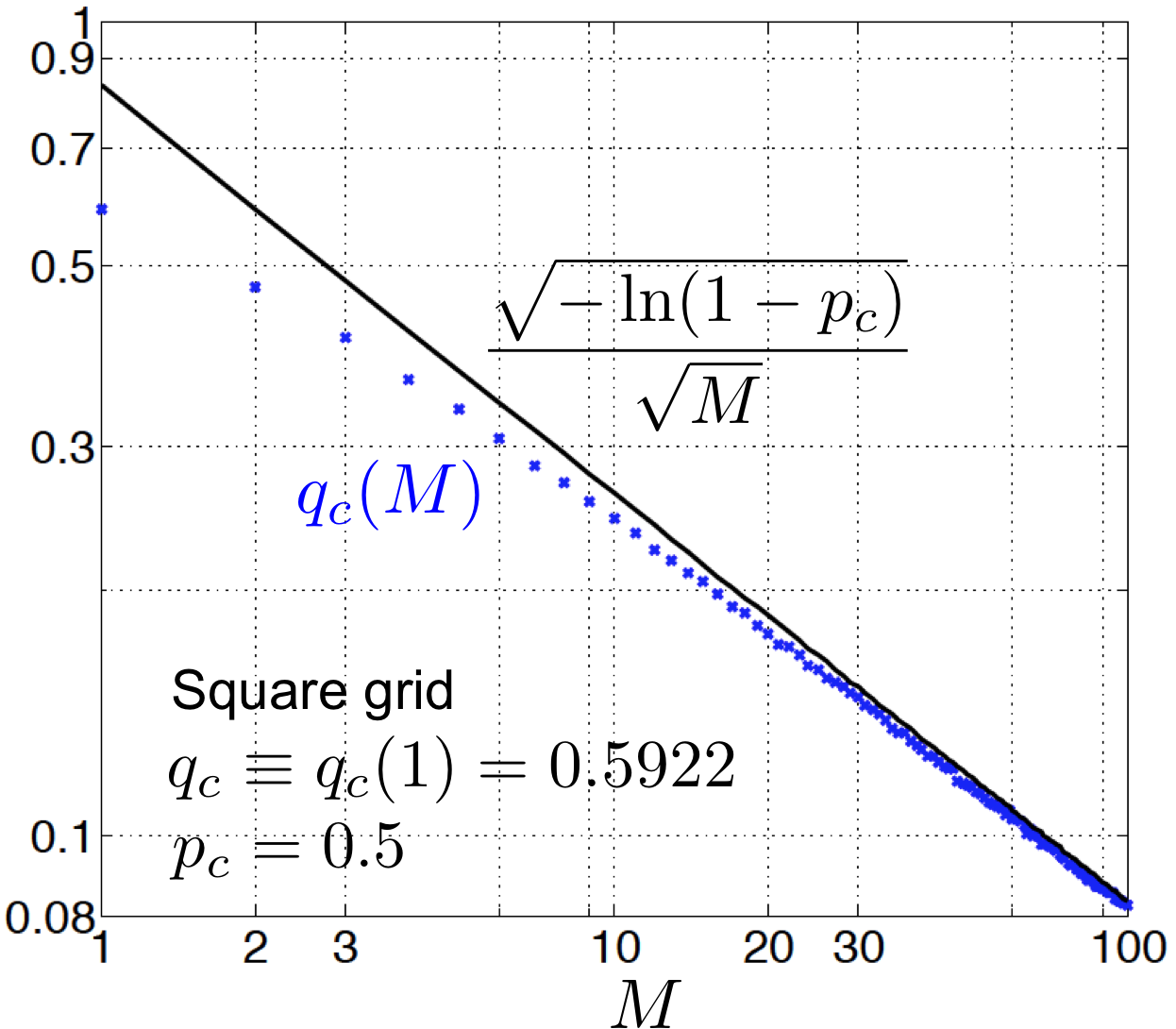}
\caption{Multilayer percolation thresholds for a square grid (blue crosses) were numerically evaluated for $M = 1, \ldots, 100$, where $q_c(1) = q_c = p_c=0.59274605079210(2)$ is the site percolation threshold~\cite{Jac15}. The upper bound $\sqrt{-\ln(1-p_c)}/\sqrt{M}$, plotted in black (solid), is a function only of the bond percolation threshold $p_c = 0.5$, and is an asymptotically tight bound in the large $M$ limit.}
\label{fig:grid_UB_intuition}
\end{figure}

The multilayer (random) graph $G^{(M)}$ is completely specified by the underlying connectivity graph $G$, the number of layers $M$, and the site-occupation probability $q$ (for each site in each layer). It is simple to see that the induced marginal probability, $p$, of any given bond in $G^{(M)}$ to be active, is given by $p = 1 - (1-q^2)^M$. Now, suppose we choose $q = c/\sqrt{M}$, with $c$ a constant. In other words, each node is occupied, on an average, in $c\sqrt{M}$ layers. Then, in the limit $M \to \infty$, we have $p = 1-\lim_{M \to \infty}(1-q^2)^M = 1 - e^{-c^2}$. If the bond activation events were statistically independent, then for $p \ge p_c$, where $p_c$ is the bond percolation threshold of $G$, a giant connected component (GCC) would appear in $G^{(M)}$. Note here that $p \ge p_c$, in the $M \to \infty$ limit, is equivalent to $c \ge \sqrt{-\ln(1-p_c)}$. However, the bond activation events in $G^{(M)}$ are not independent. They have a positive spatial correlation. In other words, one bond being active makes it more likely for its neighboring bond to be active. Since $p$ is the fractional size of the edge set in the underlying graph that is active, introducing positive bond-to-bond nearest-neighbor spatial correlations, for a given $p$, implies that the active bonds will be closer together, and hence $p > p_c$ should be more than sufficient for a GCC to exist in $G^{(M)}$, thereby making $\sqrt{-\ln(1-p_c)}/\sqrt{M}$ an upper bound to the multilayer percolation threshold $q_c(M)$, for any finite $M$. Note however that the above argument (of why percolation should happen at a strictly lower value of $p$ for a correlated bond process) is not rigorous. We conjecture however that it holds true for the particular correlated bond process induced by our multilayer percolation model on an arbitrary graph $G$ that has a well defined non-trivial i.i.d. bond percolation threshold. 

Clearly, for $M=1$, $G^{(M)}$ is a simple site-percolation instance over $G$ with site-occupation probability $q$. Hence $q_c(1) = q_c$ is the site-percolation threshold of $G$. In Fig.~\ref{fig:grid_UB_intuition}, we plot the numerically-evaluated values of $q_c(M)$ as a function of $M$ for a square grid. Note that as $M$ grows large, the aforesaid upper bound gets progressively tighter. This indicates that the reason which caused in the first place $\sqrt{-\ln(1-p_c)}/\sqrt{M}$ to be an upper bound (and not equal) to $q_c(M)$---that the bond-activation events of $G^{(M)}$ being positively correlated---dwindles away in the large $M$ limit. One can show that this is indeed true. In other words, if the site-occupation probability is chosen to be $q = c/\sqrt{M}$ in our multilayer graph construction, for $M$ large, the induced bond-activation events on the multilayer graph $G^{(M)}$ approach towards being statistically independent. Therefore, in this limit, $G^{(M)}$ resembles an i.i.d. bond-percolation instance of $G$. Thus when $p > p_c$, the bond-percolation threshold of $G$, a GCC appears in $G^{(M)}$. Therefore, $q \sim \sqrt{-\ln(1-p_c)}/\sqrt{M}$ in the $M \to \infty$ limit, showing that the upper bound is asymptotically tight. In Theorem~\ref{thm:multilayersite_conj} of Appendix~\ref{app:exact}, we provide a rigorous proof of the independence of the bond-activation events of $G^{(M)}$ in the $M \to \infty$, albeit only for the case when $G$ is a tree. We conjecture (and have ample numerical evidence in its favor), that this fact about the asymptotic independence of bond activation events (of $G^{(M)}$ for $q = c/\sqrt{M}$ in the $M \to \infty$ limit) holds true for any arbitrary graph $G$ that has a well defined and non-trivial bond-percolation threshold $p_c$.

Further, recalling that the marginal probability of each bond's activation satisfies $p = 1 - (1-q^2)^M$, it is simple to see that in the $M \to \infty$ limit, if the site-occupation probability $q$ is chosen to be any function of $M$ that diminishes any faster than $1/\sqrt{M}$, then all the bonds of $G^{(M)}$ are inactive with high probability (w.h.p.), whereas if $q$ is chosen as any function of $M$ that diminishes even a little slower compared to $1/\sqrt{M}$, then all the bonds of $G^{(M)}$ are active w.h.p., thereby showing that the $1/\sqrt{M}$ scaling of $q$ is a sharp connectivity threshold.

For $M=1$, the multilayer model is equivalent to standard i.i.d. site percolation, and as $M$ increases, even though our multilayer construction is inherently a site-based model (i.e., defined and driven solely by site occupations), the site-occupation thresholds for long-range connectivity are determinable only from the bond-percolation threshold of the underlying connectivity graph when the number of layers grows large. The thresholds $q_c(M)$ decrease as $M$ increases, but they decrease as $c/\sqrt{M}$, where the constant $c$ is only a function of $p_c$, the bond-percolation threshold, when $M$ is large. Therefore, bond percolation naturally emerges from a multi-layer site percolation problem. This in turn suggests that there may be a deeper connection of this problem to the traditional site-bond (or, mixed) percolation problem, which is a model that interpolates between the standard site and bond percolation models in a natural way, with each site occupied independently with probability $q$ and each bond activated independently with probability $p$. In Section~\ref{sec:sitebond}, we will make this connection quantitatively rigorous, and will show how one can translate known results on the critical region for the site-bond percolation problem to the multilayer problem, and vice versa.

Finally, let us see how one can further tighten the upper bound to $q_c(M)$ discussed above. The probability that any given bond in $G^{(M)}$ is active, $p = 1 - (1 - q^2)^M$. The bonds in $G^{(M)}$ can be thought of as generated by a positively-correlated bond-percolation process on $G$ with a marginal bond probability $p$. Hence, $p > p_c$ is sufficient for a GCC to appear in $G^{(M)}$, where $p_c$ is the bond percolation threshold of $G$. It thus follows that $q \ge \sqrt{1-(1-p_c)^{1/M}}$ is sufficient for percolation. Therefore the multilayer threshold must satisfy, $q_c(M) \le \sqrt{1-(1-p_c)^{1/M}}$. One can further upper bound the right-hand side of the above as $\sqrt{1-(1-p_c)^{1/M}} \le \sqrt{-\ln(1-p_c)}/\sqrt{M}$, $\forall M \ge 1$, to obtain, $q_c(M) \le \sqrt{-\ln(1-p_c)}/\sqrt{M}$. In Appendix~\ref{app:intuitive}, we provide an additional intuition behind the above upper bound.

\section{Analytical results for a multilayer random graph}\label{sec:random}

In this section, we will consider multilayer percolation on large random graphs with an arbitrary, but known, node-degree distribution. This model of random graph analysis---based on a probability-generating function (PGF) approach---is known as the {\em configuration model} (CM)~\cite{New00, New01}. The CM has found use in modeling several real life networks that have non-Poisson distributions, such as the truncated power-law and exponential distributions~\cite{Wat98, New00a}. Newman studied a multilayer random graph model, which is related but different from ours~\cite{New03}. 

Let $p_k$ denote the probability that a randomly selected node has degree $k$. Let $\cM$ denote the set of $M$ layers and let $q_k$ denote the probability that a node of degree $k$ is occupied in layer $m \in \cM$. As before, the events that a node is occupied in different layers are assumed to be independent. $p_k (1-(1-q_k)^M)$ is the probability of a node having degree $k$ and being occupied in at least one layer, and
\[ F_0(x) = \sum_{k=0}^\infty p_k (1-(1-q_k)^M) x^k \]
is the PGF of this distribution. Let us follow a randomly chosen edge $e(u, v)$ starting from a node $u$, occupied in $n \le M$ layers, to node $v$. Node $v$ has degree distribution proportional to $kp_k$~\cite{New00}. Thus the PGF of the distribution of $v$ having degree $k$ {\em and} being occupied in at least one of $n$ layers that $u$ is occupied in, is given by
\[ F_n(x) = \frac{\sum_{k=1}^\infty kp_k(1-(1-q_k)^n)x^{k-1}}{z}\]
for $1 \le n \le M$, where $z=\sum_k kp_k$ is the average node degree.

Let $H_0(x)$ denote the PGF of the cluster size that a randomly selected node belongs to. It is easy to argue that,
\begin{equation} \label{eq:H0}
H_0(x) = 1-F_0(1) + x\sum_{k=0}^\infty p_k \sum_{l=1}^M \binom{M}{l}q_k^l(1-q_k)^{M-l} H_l(x)^k. \end{equation}
Let us define $H_n(x)$ as the PGF for the size of a cluster that a neighbor of the node belongs to provided that it is occupied in at least one of the $n$ layers in which the randomly selected node is occupied in. It is given by
\begin{align*}
H_n(x) 
& = 1-F_n(1) + x \sum_{k=1}^\infty \frac{kp_k}{z}\\
& \;\times \;\sum_{l=1}^M \left[\binom{M}{l}-\binom{M-n}{l}\right] q_k^l(1-q_k)^{M-l}H_l(x)^{k-1} 
\end{align*} 
where $\binom{j}{i}$ is defined to be zero whenever $i>j$. The combinatorial term in the inner sum corresponds to the number of combinations of $l$ layers at a neighbor that overlap the $n$ layers of the original node. When $l > M-n$, all possible combinations of $l$ layers at the neighbor overlaps the $n$ layers in the original node, yielding $\binom{M}{l}$ whereas when $l \le M-n$, we have to subtract out the number of $l$ layer combinations that do not overlap the $n$ layer combinations at the original node, $\binom{M-n}{l}$. We are interested in the average cluster size that a randomly selected node belongs to, which is given by $\mu_0 = H'_0(1)$,
\begin{equation}
\mu_0 = F_0(1) + x\sum_{k=0}^\infty kp_k \sum_{l=1}^M \binom{M}{l}q_k^l(1-q_k)^{M-l} \mu_l
\label{eq:mu0}
\end{equation}
where $\mu_n = H'_n(1)$, i.e.,
\begin{align*}
\mu_n &= F_n(1) + \sum_{k=1}^\infty \frac{(k-1)kp_k}{z}\\
&\;\times\;\sum_{l=1}^M \left[\binom{M}{l}-\binom{M-n}{l}\right] q_k^l(1-q_k)^{M-l}\mu_l .
\end{align*}
Consider the case $q_k=q$. We introduce the matrix $\bA(q) = [A_{ij}]$ with 
\begin{equation}
A_{ij} = 
\begin{cases}
-C\Bigl(\binom{M}{j}-\binom{M-i}{j}\Bigr) q^j(1-q)^{M-j} & i\ne j \\
1-C\Bigl(\binom{M}{j}-\binom{M-j}{j}\Bigr) q^j(1-q)^{M-j} & i=j
\end{cases}
\label{eq:Adefinition}
\end{equation}
where $C = \sum_k(k-1)kp_k/z$.  Now
$\mu = \left[\mu_1,\ldots , \mu_M\right]$ is a solution of
\begin{equation}
\bA (q)\mu^T = \bb^T
\label{eq:AmuT}
\end{equation}
where $\bb = (F_1(1), \ldots F_M(1))$. 
We define the critical occupancy probability, $q_c(M)$, such that when $q > q_c(M)$, there exists one infinite size spanning cluster (or, GCC) w.h.p., and when $q < q_c(M)$, there exist only finite size clusters.
\begin{figure}
\centering
\includegraphics[width=\columnwidth]{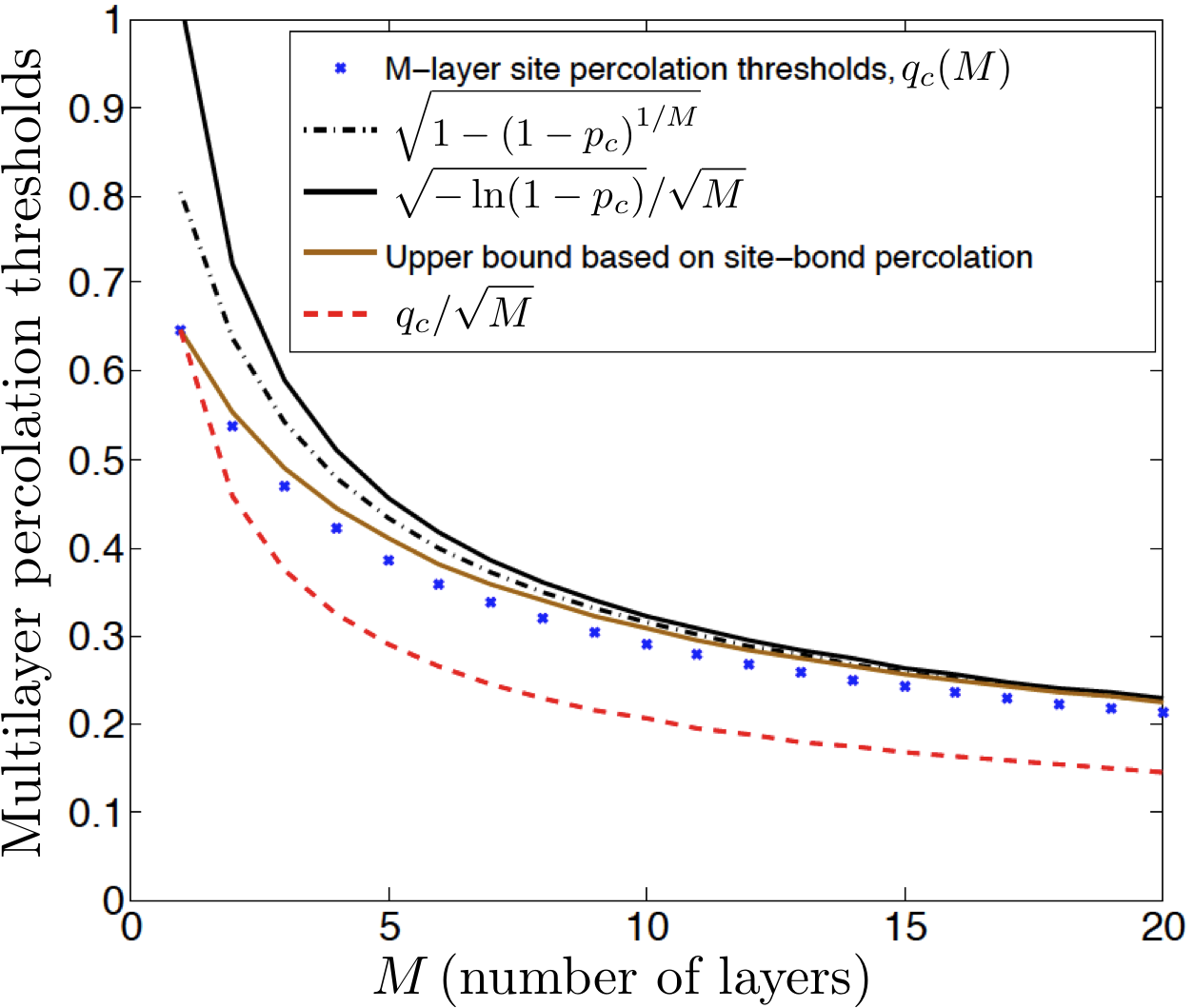}
\caption{Upper and lower bounds on the multilayer site-percolation thresholds for a random graph with a truncated power law degree distribution $p_k = 0, k = 0$, $p_k = Ck^{-\tau}e^{-k/\kappa}, k \ge 1$, with $\kappa = 10$ and $\tau = 2.5$. The exact $q_c(M)$ values were obtained by solving Eq.~\eqref{eq:detAq}.}
\label{fig:randomgraph_thresholds}
\end{figure}
Assume that $C>1$ (which is required for $G$ to have a GCC w.h.p. even with all nodes and bonds occupied). The size of a GCC is a constant fraction of the size of the graph. Hence, for an infinite random graph, appearance of a GCC in the $M$-layer graph is equivalent to $\mu_0$ diverging (to infinity). As Eq.~\eqref{eq:mu0} shows, $\mu_0$ is a constant plus a linear combination of $\left\{\mu_k\right\}$, $k=1, \ldots, M$. Hence, at least one element of $\mu$ must diverge for $\mu_0$ to diverge. Since $\boldsymbol b$ is a constant vector, this can only happen if ${\boldsymbol A}$ is singular. Thus, $q_c(M)$ is given by the solution of
\begin{equation}
\det (\bA (q)) = 0
\label{eq:detAq}
\end{equation}
within the interval $[0,1]$. Numerically, it is easy to verify that $\det (\bA (q))$ has a unique zero in $[0,1]$ for all $M \ge 1$ as long as $C > 1$. We conjecture that this is always true. In Appendix~\ref{app:random}, we provide a rigorous proof of the fact that $q_c(M)$ is given by the smallest solution of $\det (\bA (q)) = 0$ in the interval $(0, 1)$.

For the case $M=1$, this corresponds to finding the solution of $1 - qC = 0$, which yields the known result $q_c = 1/C$~\cite{New01}. For the case of $M=2$,  $q_c$ is the unique solution of the following polynomial
\begin{equation} \label{eq:det-2}
  C^2q^4-C^2q^3-Cq+1 = 0
\end{equation}
within the interval $[0,1]$, which is
\begin{equation}
q_c{(2)} = \frac14\left[a+1-\sqrt{3+2a-a^2}\right]
\end{equation}
with $a = \sqrt{1+8/C}$. It is easy to show that there exists at least one real root in the interval $[0,1]$ provided that $C>1$ as $\det (\bA (0)) = 1$ and $\det (\bA (1)) = 1-C < 0$. 
\begin{figure}
\centering
\includegraphics[width=\columnwidth]{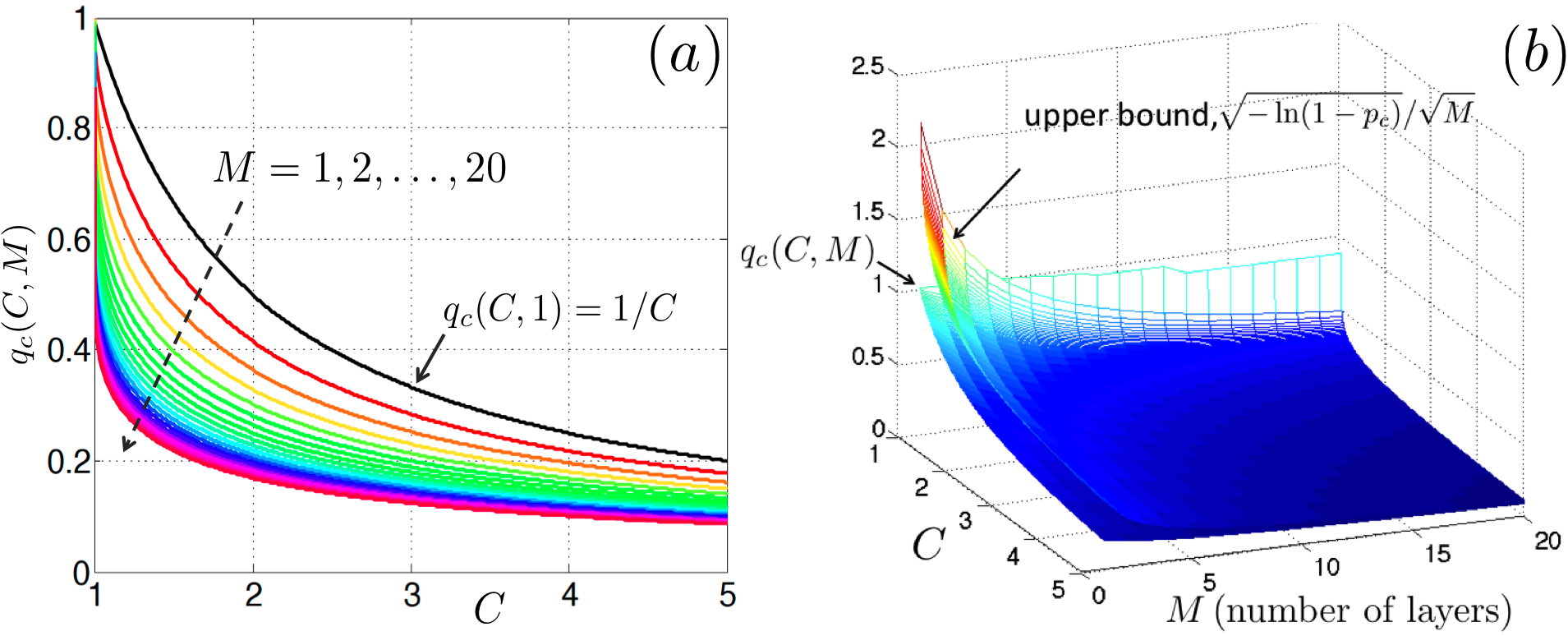}
\caption{(Color online) (a) The multilayer threshold $q_c(C, M)$ as a function of $C$ for $M=1, 2, \ldots, 20$ layers. (b) Comparing $q_c(C,M)$ with $\sqrt{-\ln(1-p_c)}/\sqrt{M}$ (with $p_c = 1/C$), which is seen to be an upper bound to $q_c(C,M)$ as argued in Appendix~\ref{app:intuitive}. The random graphs used for these evaluations were chosen from a truncated power law node-degree distribution $p_k = 0, k = 0$, $p_k = Ck^{-\tau}e^{-k/\kappa}, k \ge 1$, with $\kappa = 10$ and $\tau = 2.5$.}
\label{fig:randomgraph_results}
\end{figure}

\begin{figure}
\centering
\includegraphics[width=\columnwidth]{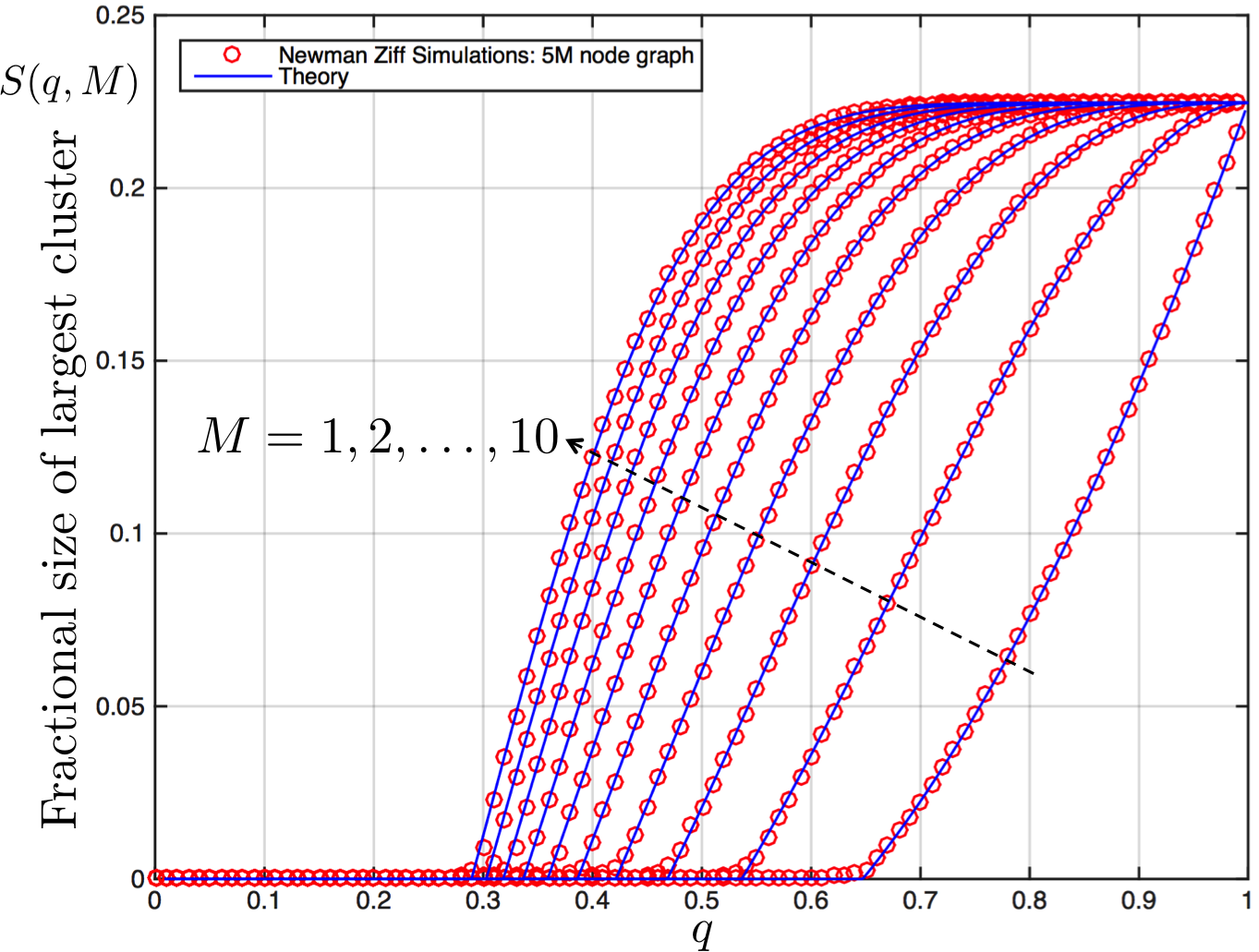}
\caption{(Color online) Size of the spanning cluster for multilayer percolation as a function of the single-layer site-occupation probability $q$ for a random graph with a truncated power law node-degree distribution; $p_k = 0, k = 0$, $p_k = Ck^{-\tau}e^{-k/\kappa}, k \ge 1$, with $\kappa = 10$ and $\tau = 2.5$. The theory plots were obtained by solving Eqs.~\eqref{eq:SqM_formula} and~\eqref{eq:fixedpoint}. The numerical plots were obtained by Newman-Ziff style simulations, via averaging over $10$ instances of a $5$ million node random graph.}
\label{fig:randomgraph_clustersize}
\end{figure}
In the supercritical regime, there is one infinite size cluster and many small finite size clusters. The PGF of the size of a small cluster is given by $H_0(x)/H_0(1)$  with $H_0(x)$ given by (\ref{eq:H0}).   The average size of these clusters is $\mu_0/H_0(1)$ with $\mu_0$ given by (\ref{eq:mu0}). Finally, the fractional size of the giant connected component is given by $S = 1 - H_0(1)$. For the $M$-layer random graph, the fractional size of the giant connected component is given by 
\begin{equation}
S(q,M) = 1 - (1-q)^M - \sum_{k=0}^\infty p_k \sum_{l=1}^M \binom{M}{l}q^l(1-q)^{M-l}u_l^k
\label{eq:SqM_formula}
\end{equation}
where $\boldsymbol{u} = [u_1, u_2, \ldots, u_M]^T \in [0, 1]^M$ is given by the following self-consistency matrix equality
\begin{equation}
\boldsymbol{u} = \boldsymbol{s} + \boldsymbol{B}\boldsymbol{v}
\label{eq:fixedpoint}
\end{equation}
where $s_l = (1-q)^l$, 
$
v_l = f(u_l) \equiv \sum_{k=1}^\infty ({kp_k}/{z})u_l^{k-1},
$
and $B_{ij} = \left[\binom{M}{j} - \binom{M-i}{j}\right]q^j(1-q)^{M-j}$.

In order to verify our theory, we performed numerical simulations of layered site percolation on random graphs with up to $M=20$ layers, and $5$ million nodes, and compared with results obtained from the theory we developed above. We chose random graphs with node degrees distributed according to the truncated power law~\cite{Cal00},
\begin{equation}
p_k = \left\{
\begin{array}{ll}
0 & {\text{for}}\,k=0,\\
Dk^{-\tau}e^{-k/\kappa} & {\text{for}}\,k \ge 1,
\end{array}
\right.
\label{eq:cluster}
\end{equation}
where $D = \left[{\rm Li}_\tau(e^{-1/\kappa})\right]^{-1}$ is a normalization constant, with the polylogarithm function, ${\rm Li}_s(x) \equiv \sum_{k=1}^\infty {x^k}/{k^s}$. We chose this distribution for our simulations since it is seen in a number of real-world social networks including collaboration networks of movie actors~\cite{Ama00} and
scientific collaborations based on co-authorship of publications~\cite{New04}. The pure power-law distributions seen in Internet data are also included as a special case $\kappa \to \infty$~\cite{Fal99}. 

For our numerical evaluations of $q_c(M)$, we used the Newton-Raphson method to extract the unique root of Eq.~\eqref{eq:detAq}. Fig.~\ref{fig:randomgraph_thresholds} shows an example calculation of $q_c(M)$ for a truncated power law node degree distribution, and various bounds to it discussed earlier. In order to solve for the largest cluster size $S(q,M)$, we solved Eq.~\eqref{eq:fixedpoint} numerically using a multi-dimensional iterative fixed-point method to search for the unique solution of ${\boldsymbol u} \in [0, 1]^M$. One interesting thing to note is that the multilayer threshold $q_c(C, M)$ is only a function of $M$ and $C \equiv \sum_{k=1}^\infty (k-1)kp_k/z$, regardless of the actual distribution $\left\{p_k\right\}$. In Fig.~\ref{fig:randomgraph_results}, we plot $q_c(C, M)$ for different values of $C$ and $M$. The evaluations of the largest cluster size $S(q,M)$ as a function of $q$---both using the solution of Eq.~\eqref{eq:SqM_formula} as well as using efficient Newman-Ziff style Monte-Carlo simulations on random graph instances with $5$ million nodes---for a truncated power law node-degree distribution, with $\kappa = 10$ and $\tau = 2.5$, are summarized in Fig.~\ref{fig:randomgraph_clustersize}. Excellent agreement is seen between theory and numerical simulations.

\section{Numerical results for regular lattices}\label{sec:lattice}

\begin{figure}
\centering
\includegraphics[width=\columnwidth]{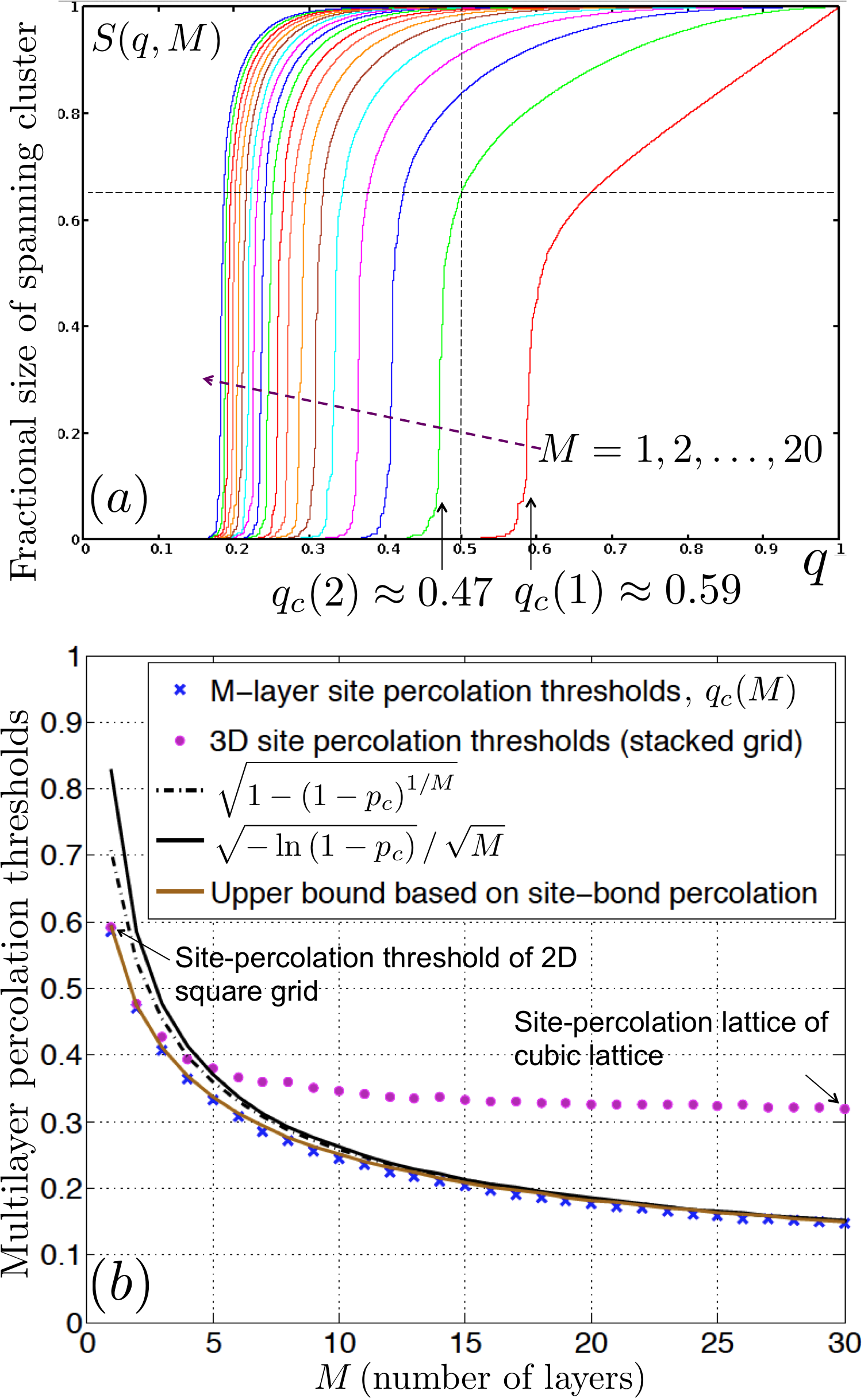}
\caption{(Color online) Thresholds and bounds for the multilayer square lattice. Simulations performed on a $262144$ node lattice.}
\label{fig:squaregrid_results_main}
\end{figure}
We numerically evaluated $q_c(M)$ for various regular lattices, including the square, triangular, kagome, and archimedian lattices. The results for a regular square grid are shown in Fig.~\ref{fig:squaregrid_results_main}, and for a regular kagome lattice in Fig.~\ref{fig:kagome_results_main}. The size of the largest component for the $M$-layer lattice exhibits the usual second-order phase transition at $q = q_c(M)$. In Section~\ref{sec:LBincorrect}, we will show a very compelling (yet, incorrect in general) argument as to why the following general lower bound to $q_c(M)$ should hold: $q_c(M) \ge q_c/\sqrt{M}$, where $q_c \equiv q_c(1)$ is the site percolation threshold. In Section~\ref{sec:kagome}, we will prove an analytical lower bound to $q_c(M)$ for the kagome lattice---adapting the Scullard-Ziff triangle-triangle transformation technique,---which is seen to be extremely close to $q_c/\sqrt{M}$.
\begin{figure}
\centering
\includegraphics[width=\columnwidth]{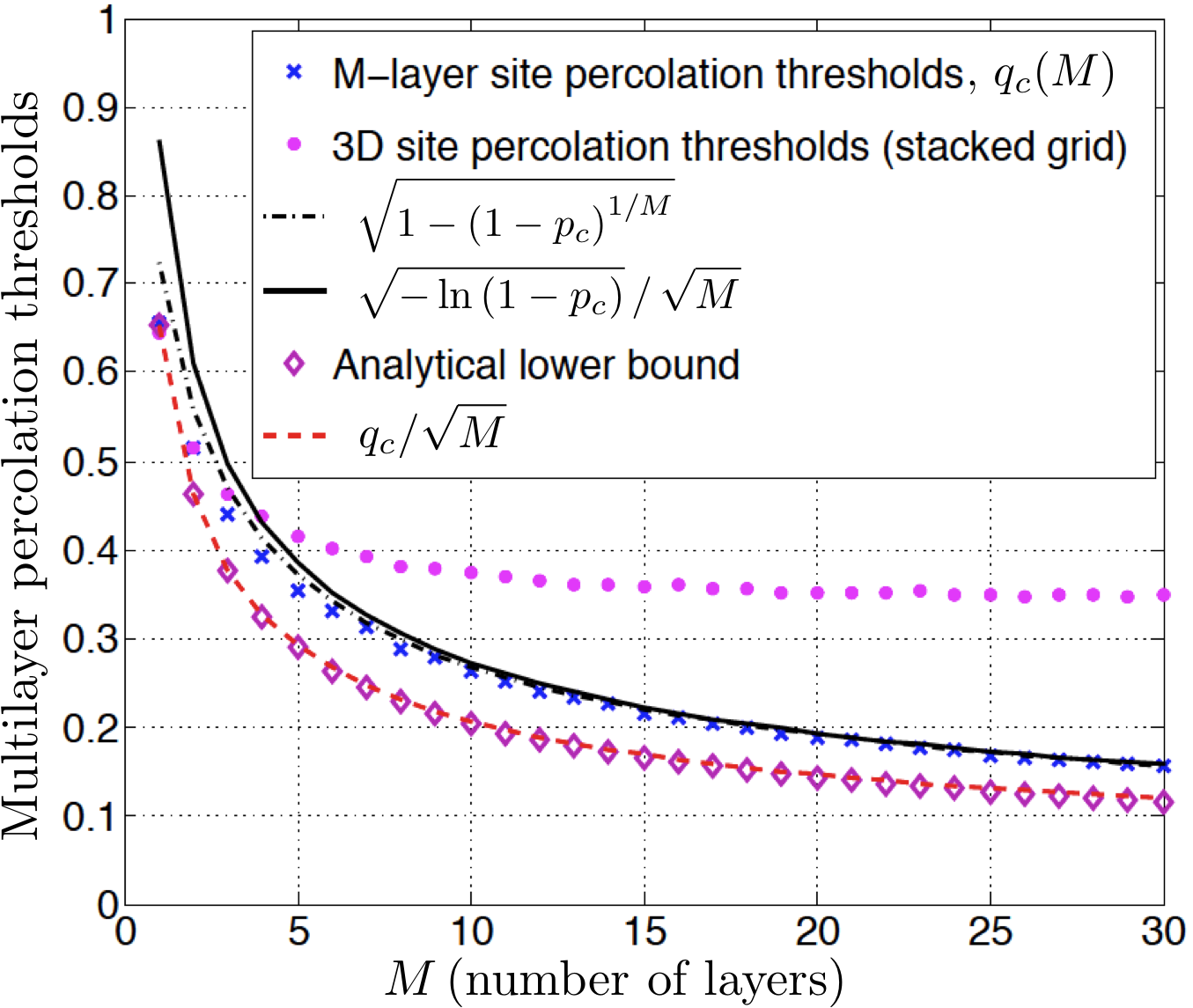}
\caption{(Color online) Thresholds and bounds for the multilayer kagome lattice. Simulations performed on a $196608$ node lattice.}
\label{fig:kagome_results_main}
\end{figure}

\subsection{An intuitive lower bound to $q_c(M)$ that holds for most regular lattices, but not in general}\label{sec:LBincorrect}

Consider an i.i.d. site-percolation process with site-occupation probability $Q$, and an $M$-layer process with single-layer site-occupation probability $q$, such that the marginal probability of a single bond to be activated in either case are identical, i.e., $Q^2 = 1 - (1-q^2)^M$. Recall now our argument above that as $M$ increases from $1$ to $\infty$, the multilayer graph $G^{(M)}$, at percolation, transitions from being identical to a pure site-percolation instance of $G$ (where bond activation events have positive spatial correlation) to a pure bond-percolation instance of $G$ (where the bond activation events are independent). Hence, one might argue that for the same total number of bonds in the respective percolating instances of a graph, if the multilayer graph percolates, that the i.i.d. site-occupied graph must also percolate (since the bond activations have higher positive spatial correlations in the latter). Thus, $\sqrt{1 - (1-q_c^2)^{1/M}} \le q_c(M)$. One can lower bound the l.h.s. by $q_c/\sqrt{M}$, thus obtaining $q_c(M) \ge q_c/\sqrt{M}$. 

The lower bound $q_c(M) \ge q_c/\sqrt{M}$ holds for various regular lattices with well-defined site-percolation thresholds~\cite{ZiffWiki}. However, the bound does break down for {\em fully-triangulated lattices}~\cite{Wie02}, and similar graph constructions where there are many more bonds connecting a smaller number of `key' sites, for which $p_c$ can be driven to zero, with $q_c$ held constant (see Fig.~\ref{fig:triangulated_results}). We conjecture that $q_c(M) \ge q_c/\sqrt{M}$ holds for all vertex-transitive graphs, which is backed by extensive numerical simulations.
\begin{figure}
\centering
\includegraphics[width=\columnwidth]{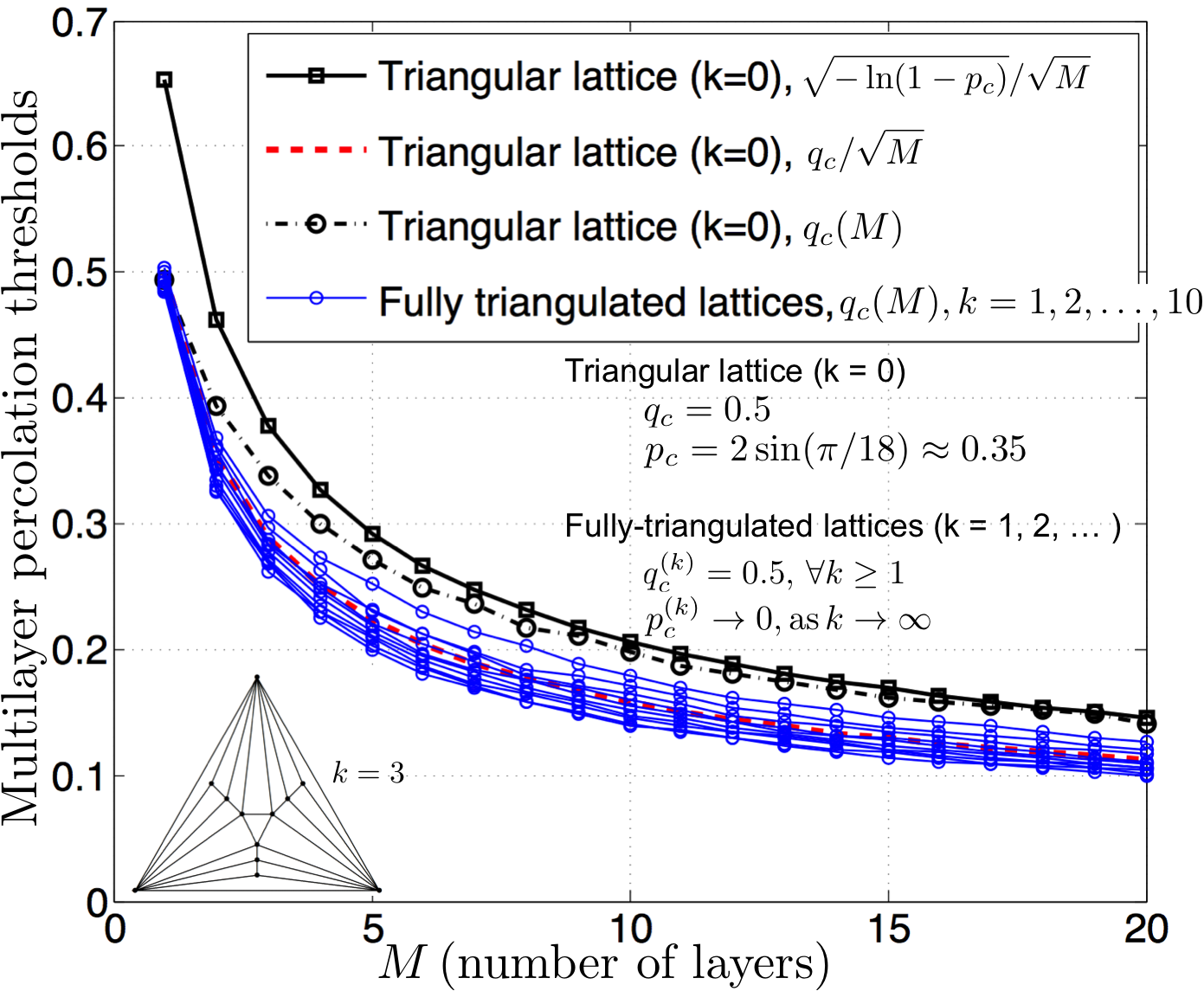}
\caption{(Color online) Plots of $q_c(M)$ for order-$k$ fully triangulated lattices for $k = 0, 1, \ldots, 10$, where $k=0$ corresponds to the simple triangular lattice. All these lattices have the same site-percolation threshold, $q_c^{(k)}=0.5$, but their bond percolation thresholds, $p_c^{(k)} \to 0$, as $k \to \infty$. The red-dashed line plots $0.5/\sqrt{M}$, therefore showing $q_c(M) \ge q_c/\sqrt{M}$ does not hold for these lattices with $k$ high enough.}
\label{fig:triangulated_results}
\end{figure}

\subsection{Lower bound on $q_c(M)$ for the Kagome lattice using the Scullard-Ziff triangle-triangle transformation}\label{sec:kagome}

Numerical evaluations of $q_c(M)$ for the kagome lattice are plotted in Fig.~\ref{fig:kagome_results_main}. For the kagome lattice, we will now prove a lower bound for $q_c(M)$ leveraging a star-triangle transformation technique developed by Scullard and Ziff~\cite{Scu06}, which was used to find exact site percolation thresholds for a large class of regular lattices. We will derive the following lower bound on $q_c(M)$:

\begin{equation}
q_c(M) \ge q_{\text{LB, Kagome}}(M)
\end{equation}
for all $M \ge 1$, where $q_{\text{LB, Kagome}}(M)$ is the unique root of the following polynomial $f_M(q)$, in $[0, 1]$:
\begin{eqnarray}
f_M(q) &=&[(1-q)(1+q-q^2)]^M + 2(1-2q^2+q^3)^M\nonumber\\
&-&[(1-q)^2)(1+2q)]^M - (1-q^2)^M   - Mq^2.
\label{eq:fMq}
\end{eqnarray}

$q_{\text{LB, kagome}}$, is seen to be extremely close, but not exactly equal, to $q_c/\sqrt{M}$, where $q_c = 1-2\sin \pi/18 \approx 0.6527$ is the site-percolation threshold of the kagome lattice~\cite{Scu06} (see Fig.~\ref{fig:kagome_results_main}). The fact that $q_{\text{LB, kagome}} \approx q_c/\sqrt{M} \le q_c(M)$ for the kagome lattice, is clearly not a coincidence, given the discussion in Section~\ref{sec:LBincorrect}.

The derivation of this bound uses a technique introduced by Scullard~\cite{Scu06}, who developed a site-to-bond transformation technique that leverages the duality of the triangular and honeycomb lattices to compute the critical surface for any correlated bond percolation process on the triangular lattice where the correlations are limited to within each triangular face~\cite{Scu06}. Fig.~\ref{fig:kagome_nz}(c) shows the setup. Imagine a triangular lattice formed by the shaded triangular faces, and for a moment ignore the dashed lines connecting the faces (i.e., collapse the three dashed lines into one node). The purpose of the dashed lines is to depict that there are no (bond or site existence) correlations in between faces. However, within each face, could there be a very complex correlated bond or site percolating network (but that network must be identical from face to face). Scullard showed that the critical condition for such a correlated-triangular lattice to percolate is given by the condition $P[A, B, C] = P[{\bar A}, {\bar B}, {\bar C}]$, where $P[A, B, C]$ is the probability that all three end nodes of a face are connected, and $P[{\bar A}, {\bar B}, {\bar C}]$ is the probability that none of the three nodes are connected to one another. A special case of this is that of correlated bond percolation, where each face has just three bonds $AB \equiv h$, $BC \equiv v$, and $CA \equiv l$, whose occupation probability is given by the joint distribution $P(h, v, l)$. The percolation condition for this case translates to:
\begin{equation}
P(v) + P({\bar v}, h, l) = P({\bar h}, {\bar l}).
\label{eq:Scullard_vlh}
\end{equation}
Scullard then observed that one way to generate such a correlated bond percolation on the triangular lattice---but one where the correlations do not traverse the lattice faces---is to consider a pure site percolating kagome lattice as shown in Fig.~\ref{fig:kagome_nz}(a), where all the orange (light) shaded triangles are the faces of a triangular lattice where the faces are detached from one another via the dashed lines as shown in Fig.~\ref{fig:kagome_nz}(c). If the site-occupation probability is $q$, it is easy to see that $P(v) = q^2$, $P({\bar h}, {\bar l})=(1-q)+q(1-q)^2$, and $P({\bar v},h,l)=0$, substituting which in~\eqref{eq:Scullard_vlh} yields a solution $q_c = 1-2\sin \pi/18 \approx 0.6527$. The last observation to be made is that if these orange (light) shaded triangular faces percolate (meaning there is a spanning cluster involving adjoining light-shaded faces), all the dashed bonds in Fig.~\ref{fig:kagome_nz}(c) in that spanning cluster must also be occupied. Reason being, due to three-point correlations, a dashed bond will be occupied with probability $1$ if two bonds on either side of it are open. More specifically, consider the bonds on triangles 2 and 3 in Fig.~\ref{fig:kagome_nz}(c). Under the transformation described above, if any one bond in each triangle is occupied, then both bounding sites on each of these bonds will be occupied. But if this is true, then it follows that the dashed bond between triangles 2 and 3 will also be occupied. Thus the two occupied bonds in the faces considered above are connected to one another via the dashed bond, just by virtue of being occupied themselves. Thus, by inserting the separating dashed triangles between the triangular faces, we have preserved the conditions for Eq.~\eqref{eq:Scullard_vlh} to be valid---that of neighboring triangular faces to be independent. Hence, $q_c = 1-2\sin \pi/18$, via this construction, is the pure site percolation threshold of the kagome lattice~\cite{Scu06}.

Now consider applying the above technique to the $M$-layer merged kagome lattice. We can still use Eq.~\eqref{eq:Scullard_vlh}, but it will only give a necessary condition for the $M$-layer lattice to percolate, since the existence of one bond each in triangles 2 and 3 will no longer necessitate the dashed bond separating them to be occupied, because the end nodes of the dashed line could now be occupied in non-intersecting layer sets. Therefore, the solution to~\eqref{eq:Scullard_vlh} will yield a lower bound to $q_c(M)$---the minimum value of single-layer site occupation probability such that the $M$-layer lattice will percolate. With a little combinatorics (detailed arguments omitted), one can calculate the following probabilities:
\begin{eqnarray}
P(v) &=& Mq^2, \\
P({\bar h},{\bar l}) &=& (1-q)^M(1+q-q^2)^M,\, {\text{and}} \\
P({\bar v}, h, l) &=& (1-q^2)^M + \left[(1-q)^2(1+2q)\right]^M \nonumber \\
&& -2(1-2q^2+q^3)^M,
\end{eqnarray}
substituting which in Eq.~\eqref{eq:Scullard_vlh}, one obtains the condition stated above to calculate the lower bound, $q_{\text{LB, Kagome}}(M) \le q_c(M)$. For completeness, we prove in Appendix~\ref{app:kagome} that $f_M(q)$ has a unique root in $(0, 1)$.

The lower bound $q_{\text{LB, kagome}}(M)$ is seen to be tantalizingly close to $q_c(1)/\sqrt{M}$ (plotted with red dashes in Fig.~\ref{fig:kagome_results_main}), but the two are not exactly equal. The magenta dots in Fig.~\ref{fig:kagome_results_main} plot the site percolation threshold $q_{c,{\rm stacked}}(M)$ of the 3D stacked kagome lattice (which is of interest due to its interesting magnetic properties~\cite{Koz12, Wan13a}). The simulations indicate that the site percolation threshold for a $50$-layer stacked lattice is roughly $q_{c,{\rm stacked}}(50) \approx 0.366$. This is in agreement with the numerically-evaluated site-percolation threshold of the infinite stacked kagome lattice, $q_{c,{\rm stacked}}(\infty) = 0.3346(4)$~\cite{van97}. 

\begin{figure}
\centering
\includegraphics[width=\columnwidth]{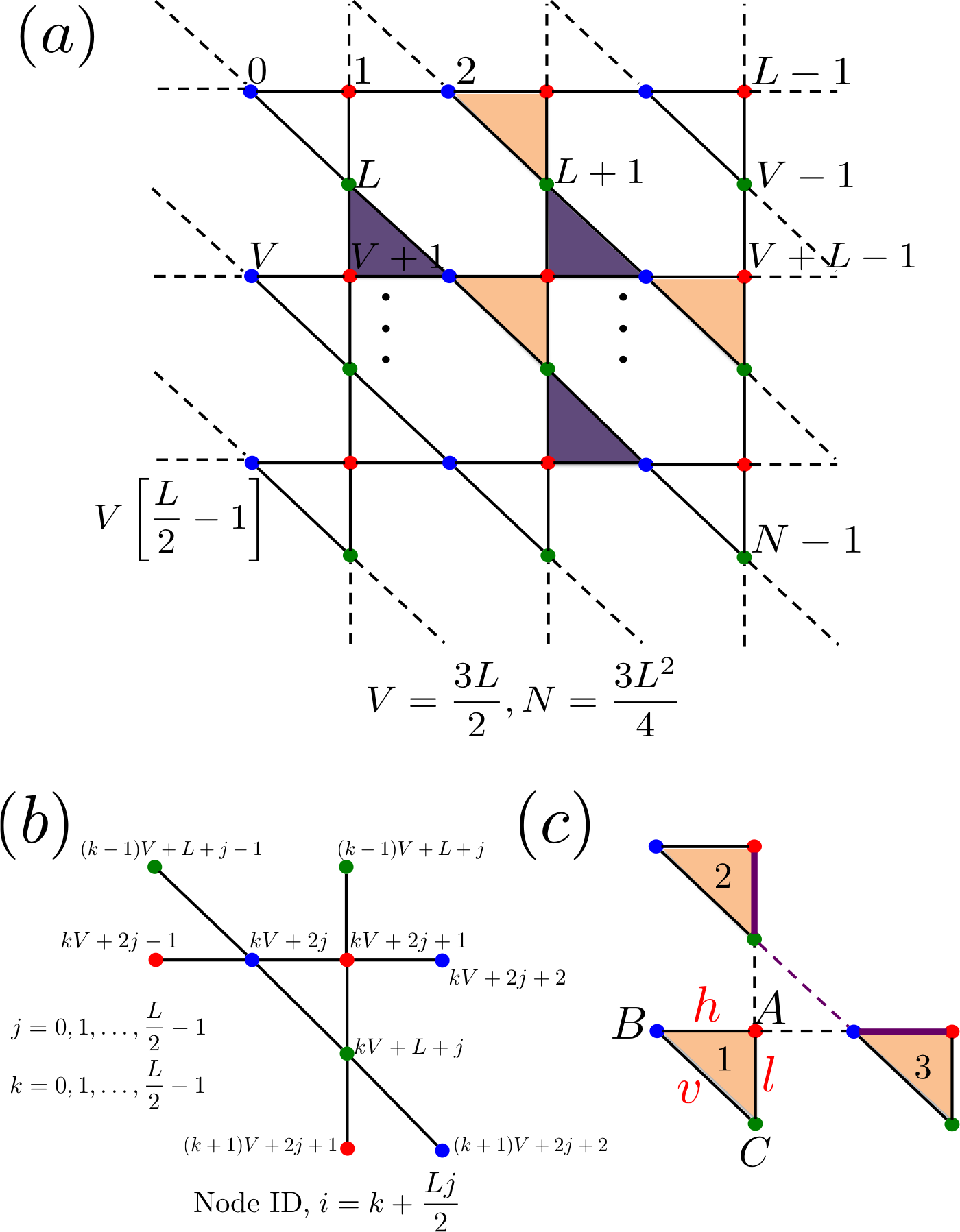}
\caption{(a) Casting the kagome lattice in the square grid, with one of every $4$ nodes in the grid removed. The dashed edges are ones that `wrap around'. (b) A `unit cell' of the kagome lattice showing the three node {\em types}, and also the node numbering convention we use to construct the nearest-neighbor matrix for use in the layered grid connectivity simulations. (c) Scullard's setup for calculating the critical region for correlated-bond percolation on a triangular lattice, with bond correlations limited within each face.}
\label{fig:kagome_nz}
\end{figure}
One interesting thing to note is that when $q > q_{\text{LB, Kagome}}(M)$, by the Scullard argument, all the orange (light) shaded triangles in the infinite multilayer Kagome lattice (see Fig.~\ref{fig:kagome_nz}(a)) will form a spanning cluster amidst themselves, i.e., assuming the purple (dark) shaded triangles do not come in the way of a pair of occupied nearest neighbor light-shaded triangles to get `connected'. But then, because of symmetry, when $q > q_{\text{LB, Kagome}}(M)$, all the purple (dark) shaded triangles should also have a spanning cluster (`ignoring' the light-shaded triangles). So, when $q$ is in the regime, $q_{\text{LB, Kagome}}(M) < q < q_c(M)$, the light-shaded triangles percolate, {\em and} the dark-shaded triangles percolate, but the full multilayer Kagome lattice does not percolate, which happens only when $q \ge q_c(M)$. This situation has some semblance with the notion of {\em explosive percolation}, that has been studied recently~\cite{Ach09}.

\section{Relationship with site-bond percolation and analytical approximations to the multilayer thresholds}\label{sec:sitebond}

As discussed above, the $M$-layer graph $G^{(M)}$ transitions from resembling site percolation to resembling bond percolation as $M$ goes from $1$ to $\infty$. Joint {\em site-bond} percolation is a well-studied extension of site and bond percolation~\cite{Ham80, Yan90, Tar99}, which is a more natural bridge between site and bond percolation, where each site is occupied and each bond is activated independently with probabilities $Q$ and $P$, respectively, and a path or a cluster can only be formed using occupied sites {\em and} activated bonds. This suggests that the two percolation models should be connected. The boundary separating the sub-critical and super-critical phases for site-bond percolation, the critical line $f_c(P,Q) = 0$, is not known exactly for any lattice. In Section~\ref{sec:translation}, we will establish a quantitative connection between site-bond and multilayer percolation, and show how one can translate the site-bond critical line $f_c(P,Q) = 0$ to an upper bound to $q_c(M)$, which is tight both at $M=1$ and $M \to \infty$. In Section~\ref{sec:approximation}, we will leverage a good approximation to the site-bond critical line to develop excellent approximations to $q_c(M)$ for general regular lattices that is only a function of the site and bond percolation thresholds $q_c$ and $p_c$, of the respective lattices. 

\subsection{Translating the site-bond critical boundary to a tight upper bound to the multilayer threshold}\label{sec:translation}

The multilayer graph can be thought of as being generated by a site-bond percolation process, where sites are independently occupied with probability $Q(q, M) = 1-(1-q)^M$, and conditioned on two nearest-neighbor sites being both occupied, the bond between them being active with probability $P(q, M) = \left[1-(1-q^2)^M\right]/\left[1-(1-q)^M\right]^2$. In other words, $P$ is the probability that two sites are occupied in at least one common layer, given they are both occupied. For $M=1$, we get $P = 1$ as expected and this reduces to pure site percolation. For $M>1$, there is one subtle difference between site-bond percolation and multilayer percolation mapped on the site-bond model as described above: the nearest neighbor bond activations have greater spatial correlation in multilayer percolation as compared to site-bond percolation, conditioned on an instance of the underlying i.i.d. site process generated with site-occupation probability $Q(q, M)$. For example, given three successive sites on a path are occupied, in the $(P, Q)$ site-bond process, the probability that both bonds between those three sites are occupied is $P^2$, whereas in multilayer percolation, the probability that both of those bonds are occupied (again conditioned on all three sites being occupied) is greater than $P^2$. This suggests that if the site-bond process on a graph $G$ percolates for a given $(P, Q)$, then for the same $(P, Q)$ value (translated to $q$ and $M$ as above), the multilayer percolation process on $G$ should also percolate. This suggests that if we know the site-bond critical line $f_c(P,Q)=0$ for a graph, and solve for $q^*(M)$ by substituting $P(q, M)$ and $Q(q, M)$ into the critical line equation, then the solution $q^*(M)$ will be an upper bound to the true multilayer percolation threshold $q_c(M)$ for that graph. If the critical line is only available numerically, we can find $q^*(M)$ by solving for the intersection of $f_c(P,Q)=0$ with $PQ^2 = 1-(1-q^2)^M$. Note that the above argument is not a formal proof that $q^*(M) \ge q_c(M)$, but we haven't found a single graph for which this upper bound is violated. The brown solid lines in Fig.~\ref{fig:squaregrid_results_main}(b) and Fig.~\ref{fig:randomgraph_results} plot this upper bound for the square grid and a random graph, respectively. This upper bound, unlike the upper bound $q_c(M) \le \sqrt{-\ln(1-p_c)}/\sqrt{M}$, is tight both at $M=1$ and $M \to \infty$, since it interpolates between pure-site and pure-bond percolation, which is a characteristic of both multilayer, and site-bond percolation.

For a random graph with degree distribution $\left\{p_k\right\}$, the critical line is a hyperbola given by $f_c(P,Q) = PQ - 1/C = 0$, with $C = \sum_k (k-1)kp_k/z$, $z = \sum_kkp_k$. The following thus readily follows:
For multilayer site-percolation on a random graph with degree distribution $\left\{p_k\right\}$, the $M$-layer thresholds satisfy, $q_c(M) \le q_{\rm UB, Random-Graph}$, where $q_{\rm UB, Random-Graph}$ is given by the unique root of the following polynomial $g_M(q)$, in $(0, 1]$:
\begin{equation}
g_M(q) = (1-q^2)^M - \frac1C(1-q)^M + (1/C) - 1.
\label{eq:randUB}
\end{equation}
See Appendix~\ref{app:sitebond_random} for proof of uniqueness of the root.

\subsection{A general approximation to $q_c(M)$ that is only a function of $p_c$ and $q_c$ of a regular lattice}\label{sec:approximation}

\begin{figure}[t]
\centering
\includegraphics[width=\columnwidth]{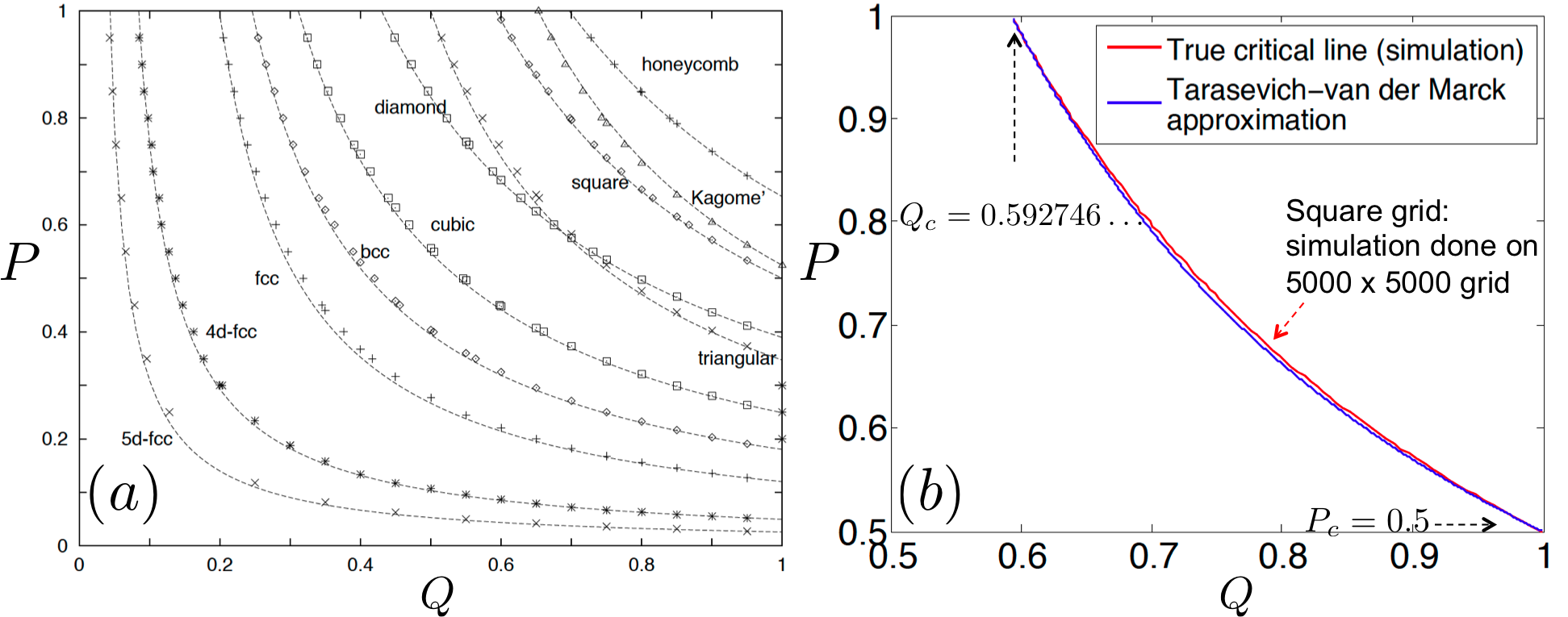}
\caption{Site-bond critical regions: comparison between true critical region and the Tarasevich-van der Marck approximation, (a) Figure from Ref.~\cite{Tar99}, (b) Refined simulations for the critical region for the square lattice using Newman-Ziff method on a $25$-million-node grid.}
\label{fig:sitebond}
\end{figure}
Yanuka and Englman proposed an approximation to $f_c(P,Q) = 0$ for regular lattices purely in terms of $q_c$ and $p_c$~\cite{Yan90}, which was later improved by Tarasevich and van der Marck~\cite{Tar99}, who showed that the critical line $f_c(P,Q) = 0$ for any lattice is well-approximated by $P(Q+A)=B$, with $A = (p_c-q_c)/(1-p_c)$ and $B = p_c(1-q_c)/(1-p_c)$ (see Fig.~\ref{fig:sitebond}). Therefore as per the discussion above, it is evident that substituting $Q = 1-(1-q)^M$ and $P = \left[1-(1-q^2)^M\right]/\left[1-(1-q)^M\right]^2$ into $P(Q+A)=B$ would result in a good approximation to $q_c(M)$ for a general lattice whose site and bond percolation thresholds ($q_c$ and $p_c$, respectively) are known. We thus have the following.

The multilayer threshold $q_c(M)$ for any graph $G$ is well-approximated by the unique solution of the following polynomial equation $f_{\rm SB}(q)=0$ in $(0, 1]$, where
\begin{equation}
f_{\rm SB}(q) = \left[1-(1-q^2)^M\right]\left[1-(1-q)^M+A\right]-B\left[1-(1-q)^M\right]^2 \nonumber
\label{eq:sitebondqc_approx}
\end{equation}
where $A(p_c, q_c)$ and $B(p_c, q_c)$ are as stated above.

\begin{figure}
\centering
\includegraphics[width=\columnwidth]{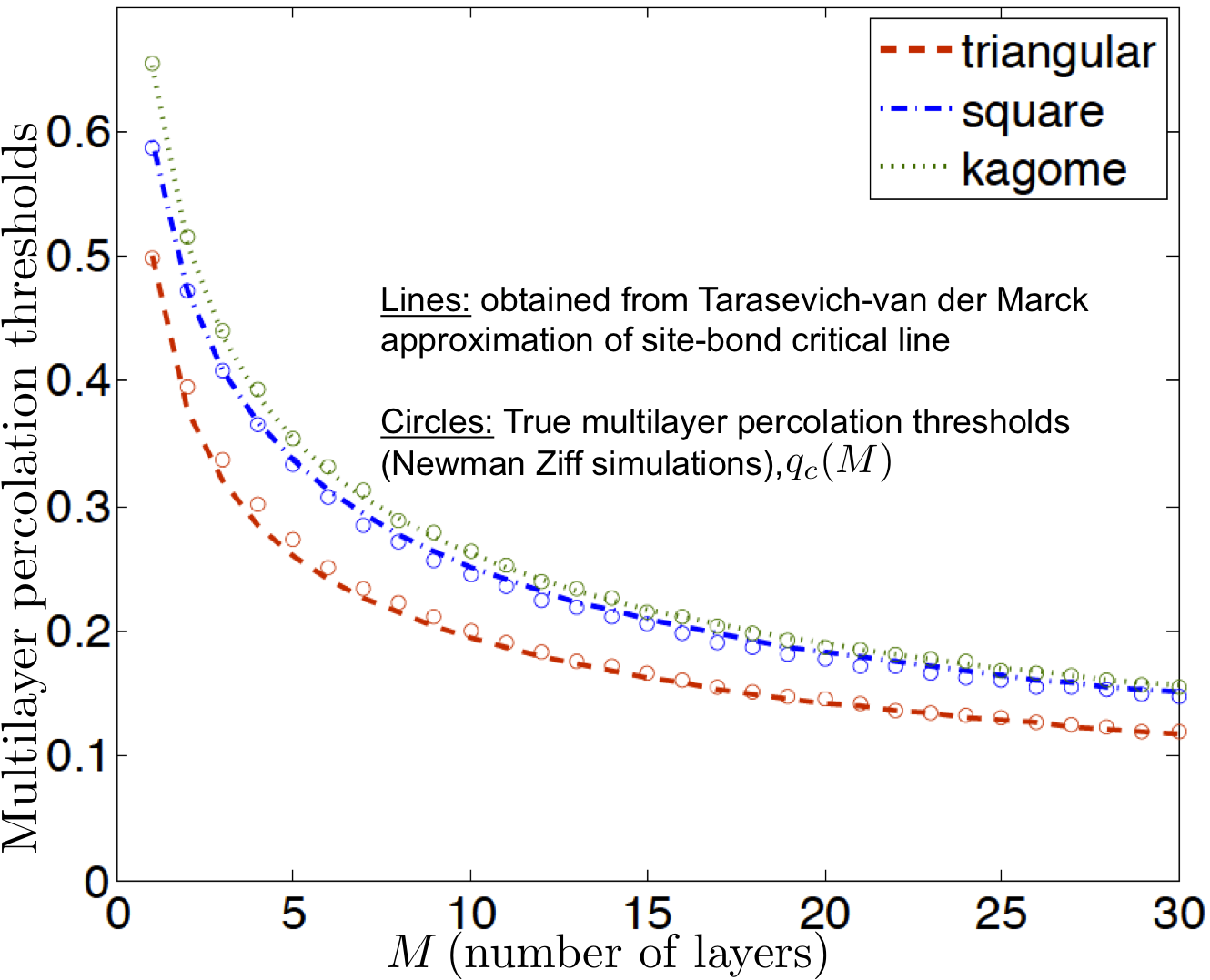}
\caption{(Color online) Comparison of $q_c(M)$---for the square, triangular and kagome lattices---to approximations obtained by using the Tarasevich-van der Marck approximations to the site-bond threshold.}
\label{fig:sitebond_thresholdapprox}
\end{figure}
The proof that $f_{\rm SB}(q)$ has a unique root in $(0, 1)$ for any given $p_c$, $q_c$ and $M$ is given in Appendix~\ref{app:approx_unique}. Fig.~\ref{fig:sitebond_thresholdapprox} shows the agreement of the approximations to $q_c(M)$ with the true thresholds $q_c(M)$ for the square, triangular and kagome lattices. Note that the approximations are neither strictly an upper nor a lower bound to $q_c(M)$ in general. The fact that the analytical approximations to $f_c(P,Q)=0$ are lower estimates of the true critical line, counters the fact that the translation of the true site-bond critical line should give us an upper bound to $q_c(M)$---thereby producing very good estimates of $q_c(M)$.

\section{Conclusions}\label{sec:conclusions}


In this paper, we studied the emergence of long-range connectivity in a specific kind of multilayer network, which have the following two properties: (1) each node in the multilayer network is a physical entity that is common to each of the $M$ layers, where the layers correspond to co-existing means of connectivity and each node may only be active in a subset of all the layers; and (2) each network layer is a subgraph of a common underlying connectivity graph $G(V, E)$, obtained by making each node in $G$ active in any given layer independently with probability $q$. The edge set $E$ defines all the possible connections the nodes in $V$ may have, some of which may remain dormant in a particular instance of the multilayer network, if the nodes an edge connects are not active in a common layer. 

We studied the properties of $q_c(M)$, the threshold value of the single-layer site-occupation probability $q$, when a spanning cluster begins to emerge in the $M$-layer network. We showed that $q_c(M) = \Theta(1/\sqrt{M})$, i.e., for the $M$-layer network to have long-range connectivity, each node must be active in $c\sqrt{M}$ layers on an average. We also showed that $c \to \sqrt{-\ln(1-p_c)}$ as $M \to \infty$, where $p_c$ is the bond percolation threshold of $G$. This arises from a realization that if each node is active in $\propto \sqrt{M}$ layers, then the induced bond activation events approach being i.i.d. as $M$ becomes large. We derived $q_c(M)$ exactly for random graphs with arbitrary degree distributions. Since $q_c(1) \equiv q_c$ is the site-percolation threshold of $G$, and the observation that $q_c(M)$ only depends upon the bond-percolation threshold $p_c$ when $M$ is large, led us to find a close relationship between the above multilayer percolation model and site-bond percolation, using which we translated a known approximation to the boundary of the site-bond critical region, to an excellent approximation of $q_c(M)$. 

One may consider various extensions of our work. For ease of analysis, we assumed the underlying population for generating each network layer to be identical, which could be relaxed in order to study a wider class of multilayer networks. Even if the assumption about the underlying connectivity graph is accurate, the occupation probability of each node need not be the same in each layer. One could consider alternative multilayer connectivity models driven by the application, such as the compatibility across communication modes or technologies, information traversal hierarchy (e.g., in military networks), causality of information flow (in temporally evolving networks) where the vertical axis in Fig.~\ref{fig:setup} would represent time, or stacked lattices (where nodes connect across nearest-neighbor layers only). Furthermore, the activity of users in each layer may evolve over time, and messages could be stored at a node and forwarded to a neighboring node at a later time instant when both nodes are simultaneously active in a common layer. It would also be interesting to analyze multilayer versions of the {\em susceptible-infected-recovered} (SIR) model of epidemic spread, such as in analyzing vaccination strategies with limited supplies when nodes can carry multiple viral strains. Finally, it may be interesting to incorporate `edge weights' (i.e., the information about how many layers an edge is active in) into the structural analysis. This will help study the `robustness' of the giant component, or that of a multilayer path. Our analysis in this paper was limited to the first-order effect---connectivity---which does not pay heed to the `strength' of an edge. Aside from robustness analysis, another interesting reason to consider edge weights would be to study multiple simultaneous inter-layer information flows. In such a scenario, if an edge that bridges two highly connected islands is only active in one layer, it may become a ``bottleneck" link, whereas being active in several layers would make the multi flow information traversal easier.

{\em Acknowledgements}---SG thanks Hari Krovi for useful discussions, and thanks Robert Ziff for providing useful comments on an early draft, and for pointing out prior results on site percolation thresholds for the stacked Kagome lattice. PN thanks Alain Jean-Marie for helping prove the uniqueness of the root for $f_{\rm SB}(q)=0$. The authors thank an anonymous referee for noting a small discrepancy in Fig.~\ref{fig:randomgraph_clustersize}, which helped identify a numerical issue in our Newman-Ziff simulations. This research was sponsored by the U.S. Army Research Laboratory (ARL) and the U.K. Ministry of Defense (MoD) and was accomplished under Agreement Numbers W911NF-06-3-0002.l (US ARL and UK MoD NIS-ITA) and W911NF-09-2-0053 (US ARL NS-CTA). CC acknowledges useful discussions with Dennis Goeckel, and support from the UMASS Amherst NSF grant CNS-1018464. This document does not contain technology or technical data controlled under either the U.S. International Traffic in Arms Regulations or the U.S. Export Administration Regulations.

\appendix

\section{The constant $c$ in the asymptotic scaling law for the multilayer site-percolation threshold, $q_c(M) \sim c/\sqrt{M}$}\label{app:exact}

In Section~\ref{sec:exact}, we showed that for our multilayer network defined on any underlying graph $G$, the threshold on the single-layer site-occupation probability which when exceeded makes a GCC appear in the $M$-layer network $G^{(M)}$, satisfies $q_c(M) \sim c/\sqrt{M}$ for $M$ large. We also argued that if $q$ is chosen to be any function that diminishes faster than $1/\sqrt{M}$, then all the bonds of $G^{(M)}$ are inactive with w.h.p., whereas if $q$ is chosen as any function of $M$ that diminishes even a little slower compared to $1/\sqrt{M}$, then all the bonds of $G^{(M)}$ are active w.h.p., thereby showing that the $1/\sqrt{M}$ scaling of $q$ is a sharp connectivity threshold for the $M$-layer network. In Section~\ref{sec:exact}, we also presented an intuitive argument to show that $c = \sqrt{-\ln(1-p_c)}$, where $p_c$ is the bond percolation threshold of $G$. In this Appendix, we will prove this rigorously for the case when $G$ is a tree. We believe that this result is true for an arbitrary graph $G$ (i.e., even one that has cycles) as long as $G$ has a well-defined bond percolation threshold. We leave the extension of the proof below, for this general case, for future work.

\begin{theorem}\label{thm:multilayersite_conj}
{\bf [Multilayer site percolation: the constant in the scaling]} For a homogenous multilayer network formed via merging $M$ random site-percolation instances of a graph $G$ with a tree topology and site-occupation probability $q$, the threshold $q_c(M)$ on the single-layer site-activation probability such that a spanning cluster appears satisfies: $q_c(M) \sim c/\sqrt{M}$ as $M \to \infty$, where $c = \sqrt{-\ln(1-p_c)}$ and $p_c$ is the bond-percolation threshold of $G$.
\end{theorem}

\begin{proof}
Given the discussion in Section~\ref{sec:exact}, the only additional argument that we need in order to complete the proof is the fact that when $q = aM^{-1/2}$ for a constant $a > 0$, the bond activation events in the $M$-layer graph $G^{(M)}$ are statistically independent, when $M \to \infty$. In Proposition~\ref{prop:tree_ind} below, we will prove this bond independence statement for the case when $G$ is a tree. Let us first begin with a few preliminaries. 

If $(Z,N,q)$ is a binomial rv with mean $\mu=Nq $ the following Chernoff's bounds hold:
\begin{equation}
\label{clower}
P(Z\leq (1- \delta)\mu) \leq e^{-\mu \delta^2/2}
\end{equation}
and 
\begin{equation}
\label{cupper}
P(Z\geq   (1+ \delta)\mu) \leq e^{-\mu \delta^2/3}
\end{equation}
for any $\delta\in (0,1)$. From the above we get
\begin{eqnarray}
P((1-\delta)\mu \leq  Z \leq  (1+\delta) \mu) 
&\geq & 1-  e^{-\mu \delta^2/3} - e^{-\mu \delta^2/2} .
\label{lower-upper}
\end{eqnarray}

For any mapping $f : \{0,1,\ldots\}\to [0,1]$, define
\begin{equation}
\label{F_f}
F_f(N)=\sum_{n=0}^N\binom{N}{n} q^n (1-q)^{N-n} f(n).
\end{equation}

\begin{lemma}
\label{lem:bounds}
For any $\delta\in(0,1)$
\begin{equation}
\label{lower-upper-F}
g_f(N,\delta) \left(1-2e^{-\mu \delta^2/3}\right) \leq F_f(N) \leq h_f(N,\delta) +  2 e^{-\mu \delta^2/3}
\end{equation}
with
\begin{eqnarray}
\label{def-gfm}
g_f(N,\delta)&:=& \min_{\lfloor{(1-\delta )\mu}\rfloor\leq n\leq \lceil{ (1+\delta )\mu} \rceil} f(n)\\
h_f(N,\delta)&:=&\max_{\lfloor{(1-\delta )\mu}\rfloor\leq n\leq \lceil{ (1+\delta )\mu} \rceil} f(n).
\label{def-hfm}
\end{eqnarray}

\end{lemma}

\begin{proof}
Fix $\delta\in (0,1)$. We have
\[
F_f(N)=S_{f,1}(N,\delta)+S_{f,2}(N,\delta)+S_{f,3}(N,\delta)
\]
with 
\begin{eqnarray*}
S_{f,1}(N,\delta)&=&  \sum_{n=\lfloor{(1-\delta )\mu\rfloor}}^{\lceil{ (1+\delta )\mu\rceil}} \binom{N}{n} q^n (1-q)^{N-n} f(n) \\
S_{f,2}(N,\delta)&=&  \sum_{n=0}^{\lfloor{(1-\delta )\mu\rfloor}-1} \binom{N}{n} q^n (1-q)^{N-n} f(n)\\
S_{f,3}(N,\delta) &=& \sum_{n=\lceil{(1+\delta )\mu\rceil}+1}^N\binom{N}{n} q^n (1-q)^{N-n} f(n) .
\end{eqnarray*}

(a) {\em Upper bound.}\\

With (\ref{def-hfm}) we get
\begin{eqnarray*}
S_{f,1}(N,\delta) &\leq &h_f(N, \delta) \sum_{n=0}^{N} \binom{N}{n} q^n (1-q)^{N-n}=h_f(N,\delta),
\end{eqnarray*}
\begin{eqnarray*}
S_{f,2}(N,\delta) &\leq & \sum_{n=0}^{\lfloor{(1-\delta )\mu\rfloor}-1} \binom{N}{n} q^n (1-q)^{N-n} \quad \hbox{since } f \in[0,1]\\
&=& P(Z \leq \lfloor{(1-\delta )\mu\rfloor}-1)\\
&\leq& P( Z\leq (1-\delta )\mu) \leq e^{-\mu \delta^2/2},
\end{eqnarray*}
by using Chernoff's bound (\ref{clower}), and 
\begin{eqnarray*}
S_{f,3}(N,\delta)&\leq &  \sum_{n=\lceil{(1+\delta )\mu\rceil}+1}^N \binom{N}{n} q^n (1-q)^{N-n} \nonumber \\
&=& P(Z\geq \lceil{(1+\delta )\mu\rceil}+1)\\
&\leq & P( Z\geq (1+\delta)\mu)\leq e^{-\mu \delta^2/3},
\end{eqnarray*}
by using Chernoff's bound (\ref{cupper}). In summary,
\begin{equation}
F_f(N)\leq h_f(N, \delta)+  2 e^{-\mu \delta^2/3}.
\label{lower-bound-F}
\end{equation}

(b) {\em Lower bound.}\\

With (\ref{def-gfm}) we get
\begin{eqnarray}
F_f(N)&\geq & S_{f,1}(N,\delta)\nonumber \\
&\geq & g_f(N, \delta)  \sum_{n=\lfloor{(1-\delta )\mu\rfloor}}^{\lceil{ (1+\delta )\mu \rceil}} \binom{N}{n} q^n (1-q)^{N-n}\nonumber \\
 &=& g_f(N, \delta)  \left( P( \lfloor { (1-\delta )\mu \rfloor} \leq Z\leq   \lceil{ (1+\delta )\mu}\rceil \right)\nonumber \\
 &\geq &g_f(N, \delta)  \left(1-e^{-\mu \delta^2/2} -  e^{-\mu\delta^2/3}\right) \quad \hbox{from (\ref{lower-upper})}\nonumber \\
 &\geq & g_f(N, \delta) \left(1-2e^{-\mu \delta^2/3}\right).
 \label{upper-bound-F}
\end{eqnarray}
Combining (\ref{lower-bound-F}) and (\ref{upper-bound-F}) yields (\ref{lower-upper-F}).
\end{proof}\\

Throughout $\bar a =1-a$ for any $a\in [0,1]$. Let us represent a tree $T$ as $T=(v,T_1,\ldots , T_L)$, where $T_l$ is a subtree hanging off the root $v$.  We denote by $v_l$ the root of subtree $T_l$.\\

Let us represent the state of a tree $X(T)$ by $X(T)= \{(X_1, X(T_1)), \ldots , (X_L,X(T_L))\}$ where  $X_l \in\{0,1\}$ is the state of  link $(v,v_l)$ between $v$ and subtree $T_l$. Here $X_l=0$ means the link is inactive, otherwise it is active. Last, let $T = \emptyset$ denote the empty tree ($L = 0$).  Suppose that there are $M$ layers and let
$q$  be the probability that a site is occupied in one layer. \\

Let $Q_{T,x(T)} = P(X(T) = x(T))$ denote the probability that the link state of the tree is $x(T)$.  It can be expressed as 
\begin{equation}
\label{QTxT}
Q_{T,x(T)}= \sum_{n=0}^M \binom{M}{n} q^n (1-q)^{M-n} Q_{T,x(T)}(n).
\end{equation}
Here $Q_{T,x(T)}(n)$ is the probability that the link state of tree $T$ is $x(T)$ conditioned on the number of occupied layers at the root being $n$. It satisfies the recursion:
\begin{eqnarray}
\lefteqn{Q_{T,x(T)}(n)= \prod_{l=1}^L \Biggl\{  \bar x_l (1-q)^n\sum_{i=0}^{M-n} \binom{M-n}{i}  q^i(1-q)^{M-n-i}} \nonumber\\
&&\times Q_{T_l,x(T_l)}(i) + x_l  \sum_{j=1}^n\binom{n}{j}q^j (1-q)^{n-j}\nonumber\\
&&\times \sum_{i=0}^{M-n} \binom{M-n}{i}  q^i(1-q)^{M-n-i}  Q_{T_l,x(T_l)}(i+j)\Biggr\}
\label{QTxn}
\end{eqnarray}
with $x(T)=((x_1,x(T_1)), \ldots, (x_L,x(T_l)))$. By convention $Q_{T_l,x(T_l)}(\cdot)=1$ if $T_l$ is only composed of the node $\nu_l$. \\

In the r.h.s. of (\ref{QTxn}) the term $ \bar x_l (\bar q^n\sum_{i=0}^{M-n} \binom{M-n}{i}  q^i(1-q)^{M-n-i} Q_{T_l,x(T_l)}(i)$ accounts for the fact that if link $(v,v_l)$ is inactive ($x_l=0$)
then  node $v_{l}$ cannot share any layer with node $v$ (this occurs with probability $(1-q)^n$) but otherwise can have any other layers among the $M-n$ remaining layers; the term $x_l \sum_{i=0}^{M-n} \binom{M-n}{i} q^i(1-q)^{M-n-i}   \sum_{j=1}^n\binom{n}{j}q^j(1-q)^{n-j} Q_{T_l,x(T_l)}(i+j)$
accounts for the fact that if link $(v,v_l)$ is active ($x_l=1$)  then node $v_l$  must share at least one layer with node  $v$ but otherwise can take any other layers. The product
accounts for the fact that subtrees $T_1,\ldots,T_L$ have stochastically independent behavior conditioned on the state, $n$, of node $v$.\\

In particular,
\begin{equation}
\label{QTxTn-height-one}
Q_{T,x(T)}(n)=\prod_{l=1}^N \left(  \bar x_l (1-q)^n  + x_l (1-\bar q^n)\right)
\end{equation}
if $T_l=(\nu_l)$ for $l=1,\ldots,L$ or, equivalently, if $T$ is only composed of its root $v$ and of its leaves $v_1,\ldots,v_L$.

Interverting the last two sums in (\ref{QTxn}) gives 
\begin{eqnarray}
\lefteqn{Q_{T,x(T)}(n)=\prod_{l=1}^L \Biggl\{  \bar x_l (1-q)^n\sum_{i=0}^{M-n} \binom{M-n}{i}  q^i(1-q)^{M-n-i}} \nonumber\\
&&\times Q_{T_l,x(T_l)}(i) + x_l \sum_{i=0}^{M-n} \binom{M-n}{i}   q^i(1-q)^{M-n-i} \nonumber\\
&& \times  \sum_{j=1}^n\binom{n}{j} q^j (1-q)^{n-j}  Q_{T_l,x(T_l)}(i+j)\Biggr\}.
\label{QTxn-al}
\end{eqnarray}

\begin{proposition}\label{prop:tree_ind}
Assume that $q=a M^{-1/2}(1+o(1))$. Then,
\begin{eqnarray}
\label{prop:asympt-tree}
&&\lim_{M\to\infty} Q_{T,x(T)}  = \prod_{l=1}^L\Biggl\{\left( \bar x_l e^{-a^2} + x_l \left(1-e^{-a^2} \right)\right)  \nonumber \\
&&\times \prod_{j=1}^{n_l} \left( \bar x_{l,j} e^{-a^2} + x_{l,j}\left(1-e^{-a^2} \right)\right)\Biggr\}
\end{eqnarray}
with $x(T)=((x_1,x(T_1)), \ldots, (x_L,x(T_L)))$,  $n_l$ the number of links in the subtree $T_l$ and $(x_{l,j}, j=1,\ldots,n_l)$ the state of these links. \\

Eq. (\ref{prop:asympt-tree}) shows that the links become stochastically independent of each other as $M$ becomes large.
\end{proposition}

\begin{proof}
Throughout we assume that $q= a M^{-1/2}(1+o(1))$. 

Consider first the  tree $T=(v, v_1,\ldots,v_N)$ of height one, composed of  the root $\nu$ and of the  leaves $\nu_1,\ldots,\nu_N$.
From  (\ref{QTxTn-height-one}) we see that
\begin{equation}
\label{height-one}
\lim_{M\to\infty} Q_{T,x(T)}(f(M))=\prod_{l=1}^L \left( \bar x_l e^{-a^2} + x_l \left(1-e^{-a^2}\right)\right)
\end{equation}
for any mapping $f$ such that $f(M)=a M^{-1/2}(1+o(1))$.\\

Let $T=((x_1,x(T_1)), \ldots, (x_L,x(T_L)))$ be an arbitrary tree, with $n_l$ the number of links in the subtree $T_l$ and $(x_{l,j}, j=1,\ldots,n_l)$ the state of these links.
We will prove that:
\begin{eqnarray}
\label{induction-tree}
&&\lim_{M\to\infty} Q_{T,x(T)}(f(M)) = \prod_{l=1}^L\left( \bar x_l e^{-a^2} + x_l \left(1-e^{-a^2} \right)\right)  \nonumber \\
&&\times \prod_{j=1}^{n_l} \left( \bar x_{l,j} e^{-a^2} + x_{l,j}\left(1-e^{-a^2} \right)\right)
\end{eqnarray}
for any mapping $f$ such that $f(M)=a M^{-1/2}(1+o(1))$.\\

We use an induction argument to prove (\ref{induction-tree}). We know from (\ref{height-one}) that (\ref{induction-tree}) is true for any tree of height one. Assume that it is
true for any tree  of height $k$ and let us prove that it is still true for a tree of height $k+1$.\\

Let $T=((x_1,x(T_1)), \ldots, (x_L,x(T_L)))$ be an arbitrary tree of height $k+1$, with $n_l$ the number of links in the subtree $T_l$ and $(x_{l,j}, j=1,\ldots,n_l)$ the state of these links. Subtrees
$(T_l)_l$ have height at most $k$ with at least one having a height of $k$.\\

Define $\mu(m)=m q$ and $\delta(m)=1/m^\alpha$ with $0<\alpha <1/4$. 
From (\ref{QTxT}), (\ref{QTxn-al}) and Lemma  \ref{lem:bounds} we obtain the following two-sided bounds 
for $Q_{T, x(T)}$ and $Q_{T,x(T)}(n)$:
\begin{equation}
Q_{T,x(T)}(a_{0}(M))(1-\gamma(M)) \leq Q_{T,x(T)} \leq Q_{T,x(T)}(a_{1}(M)) + \gamma(M)
\label{al-Qxj}
\end{equation}
with
\[
a_{0}(m):=\arg\min\left\{Q_{T,x(T)}(i) : \alpha(m)\leq i\leq \beta(m) \right\}
\]
\[
a_{1}(m):=\arg\min\left\{Q_{T,x(T)}(i) :   \alpha(m)\leq i\leq \beta(m) \right\},
\]
and
\begin{eqnarray}
&&\prod_{l=1}^l\Biggl(\bar x_l \bar q^n  Q_{T_l,x(T_l)}(b_{l,0}(M-n))  +x_l (1-\bar q^n)   (1-\gamma(M-m))^L\nonumber\\
&&\leq Q_{T,x(T)}(n)\leq 
\prod_{l=1}^L\Biggl( \bar x_l\bar q^n  Q_{T_l,x(T_l)}(b_{l,1}(M-n)) +x_l \left(1-\bar q^n\right) \nonumber\\
&&\times  Q_{T_l,x(T_l)}(c_{l,1}(M-n))+ 2 \gamma(M-n)\Biggr)
\label{al-Qxjn}
\end{eqnarray}
with
\begin{eqnarray*}
b_{l,0}(m)&:=& \arg\min\Biggl\{ Q_{T_l,x(T_l)}(i) : \alpha(m) \leq i\leq  \beta(m) \Biggr\}
\label{b0j}\\
b_{l,1}(m)&:=& \arg\max\Biggl\{ Q_{T_l,x(T_l)}(i) :  \alpha(m) \leq i \leq \beta(m)  \Biggr\}
\label{b1j}\\
c_{l,0}(m)&:=&\arg \min\Biggl\{ Q_{T_l,x(T_l)}(i+r) :  \alpha(m) \leq i \leq \beta(m),\\
&&\quad\quad 1\leq r\leq m\Biggr\}
\label{c0j}\\
c_{l,1}(m)&:=&\arg \max\Biggl\{ Q_{T_l,x(T_l)} (i+r) :  \alpha(m) \leq i \leq \beta(m), \\
&&\quad\quad1\leq r\leq m\Biggr\},
\label{c1j}
\end{eqnarray*}
where $\alpha(m):=\lfloor{(1-\delta(m))\mu(m)}\rfloor$, $\beta(m):= \lceil{(1+\delta(m))\mu(m)}\rceil$,
$\gamma(m):=2 e^{-\mu(m)\delta(m)^2/3}$.

Since $b_{l,0}(M-n)$, $c_{l,0}(M-n)$, $b_{l,1}(M-n)$ and  $c_{l,1}(M-n)$ all behave as $a\sqrt{M}(1+o(1))$  when $n=a\sqrt{M}(1+o(1))$ and $M$ is large, we 
can use the induction assumption  to replace $n$ by $a\sqrt{M}(1+o(1))$ in both the lower bound and the upper bound in (\ref{al-Qxjn}). 
By letting now $M\to\infty$ in the latter expressions  we obtain from the induction hypothesis  
that both bounds converge to  $\prod_{l=1}^L\left\{\left( \bar x_l e^{-a^2} + x_l \left(1-e^{-a^2} \right)\right)
\prod_{j=1}^{n_l} \left( \bar x_{l,j} e^{-a^2} + x_{l,j}\left(1-e^{-a^2} \right)\right)\right\}$, which proves that
\begin{eqnarray}
\label{limit-QTxTn}  
&&\lim_{M\to\infty} Q_{T,x(T)}(f(M))=\prod_{l=1}^L\Biggl\{\left( \bar x_l e^{-a^2} + x_l \left(1-e^{-a^2} \right)\right) \times \nonumber \\
&&\prod_{j=1}^{n_l} \left( \bar x_{l,j} e^{-a^2} + x_{l,j}\left(1-e^{-a^2} \right)\right)\Biggr\}.
\end{eqnarray}
From $a_{0}(M)= a\sqrt{M}(1+o(1))$ and  $a_{1}(M)=a\sqrt{M}(1+o(1))$, (\ref{limit-QTxTn}) and the bounds in (\ref{al-Qxj}), we finally get 
\begin{eqnarray}
&&\lim_{M\to\infty} Q_{T,x(T)}= \prod_{l=1}^L\Biggl\{\left( \bar x_l e^{-a^2} + x_l \left(1-e^{-a^2} \right)\right)  \nonumber \\
&&\times \prod_{j=1}^{n_l} \left( \bar x_{l,j} e^{-a^2} + x_{l,j}\left(1-e^{-a^2} \right)\right)\Biggr\},
\end{eqnarray}
which concludes the proof.
\end{proof}

This concludes the proof of Theorem~\ref{thm:multilayersite_conj}, for the case when $G$ is a tree.
\end{proof}

\begin{remark}
We believe the asymptotic independence property holds even when $G$ is an arbitrary graph and that $c$ takes the same value as above.  This is supported by extensive simulations.
\end{remark}

\begin{remark}
In Appendix~\ref{app:intuitive}, we argue (without proof) that for any (finite) $M \ge 1$, $q_c(M) \le \sqrt{-\ln(1-p_c)}/\sqrt{M}$. 
\end{remark}

Finally, it is intuitive that the cluster sizes must grow with the number of layers $M$. In other words, the single-layer site-occupation probability $q$ at which the $M$-layer network percolates (i.e., has a spanning cluster appear) should decrease as $M$ increases. The following monotonicity property on $q_c(M)$ makes this intuition precise. 

\begin{proposition}\label{prop:monotone}
$q_c(M)$ is a non-increasing function of $M$, i.e., $q_c(M) \ge q_c(M+1), \forall M$. \end{proposition}
\begin{proof}
This is easily proven using sample path arguments. Given a site-occupation probability $q$ for each of $M$ layers, the addition of the $(M+1)$-st layer with the same site-occupation probability can only increase the number of connected sites.  Consequently, if a spanning cluster appears in the network with site-occupation probability $q$, then it can only increase in size with the addition of the $(M+1)$-st layer. One practical import of this is that one can limit the search for $q_c(M+1)$ to the interval $(0,q_c(M)]$. 
\end{proof}

\section{An intuitive argument, using the coupon-collector problem, to show that: $q_c(M) \le \sqrt{-\ln(1-p_c)}/\sqrt{M}$}\label{app:intuitive}

In this Appendix, we provide an alternative intuitive argument to show that $\sqrt{-\ln(1-p_c)}/\sqrt{M}$ is an upper bound to $q_c(M)$ for {\em all} $M \ge 1$. In Section~\ref{sec:exact}, we provided one intuitive argument for the same.

In the classic coupon collector problem, one draws, with replacement, from a box containing $n$ distinct coupons. It is known that $m$ draws fetch, roughly, $n(1-e^{-m/n})$ distinct coupons. Let us say each of the $n = |E|$ bonds of $G$ is a coupon. The expected number of bonds in each layer $G_i$ is $nq^2$, since $q^2$ is the marginal probability of a bond. Therefore, each layer can be regarded as roughly $nq^2$ coupon draws. Hence $M$ layers would be seen as $m = Mnq^2$ coupon draws. $q_c$ is the value of $q$ that corresponds to the number of draws that will fetch just enough distinct coupons (bonds) for the $M$-layer graph to percolate. For standard i.i.d. bond percolation, $np_c$ distinct bonds (on an average) would be sufficient for percolation. However, since bond activations are spatially correlated in each layer, $np_c$ coupons will be more than enough for percolation. Hence we get, $np_c \ge n(1-e^{-(Mnq_c^2)/n})$, which translates to $q_c(M) \le \sqrt{-\ln(1-p_c)}/\sqrt{M}$. As the reader would notice, there are several loose ends to the above argument in mapping multilayer percolation to the classic coupon collector problem. To name some: (1) each layer does not draw {\em exactly} $nq^2$ coupons (it is an expected number); (2) furthermore, the $nq^2$ coupon draws within one layer are done {\em without replacement} (duplicate bonds can arise only from different layers); (3) a graph with exactly $np_c$ distinct bonds does not guarantee percolation. That number is the average number of bonds at percolation when each bond is drawn independently at random; and finally (4) the bonds in multilayer percolation are {\em not} drawn independently at random. Bond activations are spatially correlated. However, it is this last point, as we argue above, that leads to $\sqrt{-\ln(1-p_c)}/\sqrt{M}$ being an {\em upper} bound to $q_c(M)$, and the first three points can be dealt with using ideas similar to those used in the proof of Theorem~\ref{thm:multilayersite_conj}.

\section{Multilayer random graph: proof that $q_c(M)$ is the smallest solution of ${\rm det}(A(q))=0$}\label{app:random}

\begin{proposition}
Assume that $C>1$. For a random graph with node degree distribution $p_k$, $q_c(M)$ is the smallest solution of
\begin{equation}
\det (\bA (q)) = 0
\label{eq:detAq_app}
\end{equation}
within the interval $[0,1]$, with $A$ defined in Eq.~\eqref{eq:Adefinition}.
\end{proposition}
\begin{proof}
Since $\bA(0)$ is the identity matrix, $\det(\bA(0))=1$. On the other hand, it is easy to see that $\det(\bA(1))=1-C<0$ under the assumption that $C>1$.
Therefore, the mapping $q\to\det(\bA(q))$ has at least one zero in $[0,1]$. Let $q_0$ be such a zero. For $q$ in the vicinity of $q_0$, $\det(\bA(q))\not=0$
since $\det(\bA(q))$ is a polynomial in the variable $q$. By Cramer's rule,
\begin{equation}
\mu_n(q)=\frac{\det(\bA_n(q))}{\det\left({\bA(q)}\right)}
\label{Cramer}
\end{equation}
for $q$ in the vicinity of $q_0$ with $q\not=q_0$,   where $\bA_n(q)$ is the matrix formed by replacing the $n$-th column of $\bA(q)$ by  $\bb^T$. 
We claim that there exists at least one $n^\star \in\{1,\ldots,M\}$ such that $\det(\bA_{n^\star}(q_0))\not=0$. 
Letting $q\to q_c$, we get from (\ref{Cramer}) that $\mu_{n^\star}(q_0):=\lim_{q\to q_0} \mu_{n^\star} (q)=\infty$. Therefore,
$\mu_0(q_0)=\infty$ from (\ref{eq:mu0}) since $\mu_0(q)$ is expressed as a linear combination of $\mu_1(q),\ldots, \mu_M(q)$ with positive coefficients.
This shows that there is an infinite size spanning cluster (a giant component) when $q=q_0$.
 Via sample path arguments one can show that there exists an infinite size spanning cluster for $q>q_c(M)$, where
$q_C(M)$ is the smallest zero of $\det(\bA(q))$ in $[0,1]$.
 \end{proof}

Numerically, it is easy to verify that $\det (\bA (q))$ has a unique zero in $[0,1]$ for all $M \ge 1$ as long as $C > 1$. We conjecture that this is always true.

\section{Uniqueness of the root of $f_M(q)=0$, the analytical lower bound to $q_c(M)$ for the kagome lattice}\label{app:kagome}

In this Appendix, we provide a proof for the fact that $f_M(q) = 0$ in Eq.~\eqref{eq:fMq} has a unique solution in $(0, 1)$. Let us take $M \geq 1$. We can rewrite $f_M(q)$ as
\[
f_M(q)=3(1-2q^2+q^3)^M -(1-3q^2 +2q^3)^M -(1-q^2)^M-Mq^2.
\]
The derivative of $f_M(q)$ is $f^\prime_M(q)=2M qg_M(q)$ with
\begin{eqnarray}
g_{M}(q)&=& \frac{3}{2}(3q-4)(1-2q^2+q^3)^{M-1} \nonumber \\
&&+ 3(1-q)(1-3q^2 +2 q^3)^{M-1}\nonumber \\
&&+ (1-q^2)^{M-1}-1.
\label{g}
\end{eqnarray}
We want to show that $g_M(q)\leq 0 $ for all $q\in [0,1]$, or equivalently that the mapping $q\to f_M(q)$ is decreasing in $[0,1]$,
which will show that $f_M(q)$ has a unique zero in $[0,1]$ since $f_M(0)=1$ and $f_M(1)=-M$.\\

Since 
\[
(1-2q^2+q^3) - (1-3q^2 +2q^3)=q^2(1-q)\geq 0
\]
for all $q\in [0,1]$, we have that
\[
(1-2q^2+q^3)^{M-1}\geq  (1-3q^2 +2q^3)^{M-1}\geq 0
\]
for all $q\in [0,1]$. So, noting that $3q-4<0$, we have from (\ref{g}),
\begin{eqnarray}
g_{M}(q)&\leq& \left( \frac{3}{2}(3q-4) +3(1-q)\right)(1-3q^2 +2 q^3)^{M-1} \nonumber \\
&&+ (1-q^2)^{M-1}-1\nonumber\\
&=&-3 \left(1-\frac{q}{2}\right) (1-3q^2 +2 q^3)^{M-1} \nonumber \\
&&+ (1-q^2)^{M-1}-1
\label{up-g}\\
&\leq & 0.\nonumber
\end{eqnarray} 
The last inequality follows by observing that the first term in the right hand side of (\ref{up-g}) is always negative and $(1-q^2)^{M-1}-1\leq 0$ for all $q\in [0,1]$. Hence, the proof that $f_M(q) = 0$ has a unique root in $(0, 1)$ is complete.

\section{Uniqueness of the root of $g_{M}(q)=0$ in $(0, 1]$}\label{app:sitebond_random}

It is simple to see why $g_M(q)$ has a unique root in $(0, 1]$. We start by taking the derivative of $g_M(q)$ with respect to $q$,
\begin{eqnarray*}
g^\prime_M(q)&=&-2 q M (1-q^2)^{M-1} + \frac{M}{C}(1-q)^{M-1}\\
&=&\frac{M}{C}(1-q)^{M-1}\left( 1- 2 q C (1+q)^M\right).
\end{eqnarray*}
Let us define $h_M(q)= 1- 2 q C (1+q)^M$, so that
\begin{equation}
g^\prime_M(q)=\frac{M}{C}(1-q)^{M-1} h_M(q).
\label{g'M}
\end{equation}
We have
\begin{eqnarray*}
h^\prime_M(q)&=&-2 C(1+q)^M - 2 q C M (1+q)^{M-1}\\
&=& -2 C (1+q)^{M-1} (1+qM).
\end{eqnarray*}
Since $h^\prime_M(q)<0$ for all $q\in [0,1]$, the mapping $q\to h_M(q)$ is strictly decreasing in $[0,1]$.
From $h_M(0)=1$ and $h_M(1)=1-C2^{M+1}<0$ (since $C>1$) there exists $q_0$ such that
$h_M(q)>0$ for $q\in [0,q_0)$, $h_M(q_0)=0$ and $h^\prime_M(q)<0$ for $q\in (q_0,1]$. Hence, from (\ref{g'M}), we conclude that $g_M(q)$ is increasing in $[0,q_0)$ and
decreasing in $(q_0,1]$. Since $g_M(0)=0$ and $g_M(1)=1/C-1<0$, which shows that $g_M(q)$ has a unique zero in $(0,1]$.

\section{Uniqueness of the root of $f_{\rm SB}(q)=0$ in $(0, 1)$}\label{app:approx_unique}

In this Appendix, we will prove that $f_{\rm SB}(q)$ has a unique root in $(0, 1)$ for any given $p_c$, $q_c$ and $M$. When $M=1$, $f_{\rm SB}(q)$ has a unique zero in $(0,1)$ at $q=B-A=q_c$. Assume from now on that $M\geq 2$. The equation $f_{\rm SB}(q)=0$ is equivalent to 
\begin{equation}
\label{new-eq}
\phi(q):= \frac{1-(1-q)^M +A}{B} = \frac{(1-(1-q)^M)^2}{1-(1-q^2)^M} := \psi(q).
\end{equation}
Since $\phi^\prime(q)=M(1-q)^{M-1}/B>0$ for $q\in (0,1)$ (as $B>0$)
we conclude that  the mapping $q\to \phi(q)$ is strictly increasing in $(0,1)$.\\

Let us now show that the mapping $q\to\psi(q)$ is strictly decreasing $(0,1)$.
We find 
\begin{eqnarray}
\psi^\prime(q)&=&2 M(1-(1-q)^M)(1-q)^{M-1} \times \nonumber \\
&&\frac{\left[1-(1-q^2)^M-q(1+q)^{M-1}(1-(1-q)^M)\right]}{(1-(1-q^2)^M)^2} \nonumber \\
&=& \frac{2 M(1-(1-q)^M)(1-q)^{M-1}} {(1-(1-q^2)^M)^2)}\, \nu_M(q)
\end{eqnarray}
with $\nu_M(q):= 1- q(1+q)^{M-1}-(1-q)(1-q^2)^{M-1}$.
We have $\nu_2(q)=-q^3$. Assume that $\nu_M(q)<0$ for $q\in (0,1)$ for $M=2,\ldots,N$ and let us show that $\nu_{N+1}(q)<0$
for $q\in (0,1)$. \\

We have
\[
\nu_{N+1}(q)= \nu_{N}(q) - q^2 (1+q)^{N-1}(1-(1-q)^{N}).
\]
Since $q^2 (1+q)^{N-1}(1-(1-q)^{N})>0$ for $q\in (0,1)$ we conclude from the induction hypothesis that $\nu_{N+1}(q)<0$ for $q\in (0,1)$.
This proves that $\psi^\prime(q)<0$ for $q\in (0,1)$, which in turn shows that $\psi(q)$ is strictly decreasing in $(0,1)$.\\

Because a strictly increasing function and  a strictly decreasing function  can intersect at most once, we have proved that (\ref{new-eq}) has at most
one solution in $(0,1)$. It has exactly one solution since $\phi(0)=A/B\leq 0$ (as $A\leq 0$ and $B>0$), $\phi(1)=(1+A)/B=1/p_c> 1$, $\psi(0)=M$ and $\psi(1)=1$, which shows that $\phi(x)$ and $\psi(x)$ intersect exactly once in $(0,1)$. This completes the proof. 

\section{Adaptation of the Newman-Ziff algorithm for multilayer percolation}\label{app:NZ}

In order to run multilayer simulations on lattices and random graphs, we used an adaptation of the Newman-Ziff technique~\cite{New00}---an efficient algorithm to simulate site and/or bond percolation systems whose runtime is essentially linear (in the number of nodes or sites). This algorithm has been extensively used for numerical analyses of percolating systems, and extended to analyzing random graphs~\cite{Cal00}, continuum percolation on an Eulidean space~\cite{Mer12}, and inter-connected networks~\cite{Sch13}. We adapted the Newman-Ziff algorithm (see Appendix~\ref{app:NZ} for details) to simulate multilayer percolation for several 2D regular lattices---the square, triangular, kagome, and the family of fully-triangulated lattices~\cite{Wie02}. Let us first review the basic algorithm.

\subsection{The Newman-Ziff algorithm}
The underlying idea is based on a union-find algorithm~\cite{Tar75}. One chooses a random order in which sites (or bonds) are occupied sequentially, and the algorithm keeps track of all the connected components at each step using a union-find data structure. Each cluster is represented by one {\em root} member, and every member $i$ is linked to a unique parent $p(i)$ in the same cluster as $i$, except the roots who are their own parents. There are two main functions: {\tt findroot(i)} and {\tt merge(i,j)}. {\tt findroot(i)} follows the links (viz., $p(p(i)) \ldots$) from $i$ to the root of its cluster, $r(i)$. Every time a new member $i$ is added to the percolating system, the algorithm iterates through all nearest neighbors of the new member, and for each neighbor $j$ that is occupied, it calls the {\tt merge(i,j)} routine, which uses {\tt findroot(i)} and {\tt findroot(j)} to find $r(i)$ and $r(j)$ and declares the one whose cluster is larger to be the parent of the other (viz., $p(r(i)) = r(j)$ if $i$'s cluster is smaller than $j$'s), unless of course $r(i) = r(j)$ in which case they are already in the same cluster. The runtime of {\tt findroot(i)} is proportional to the length of the path from $i$ to $r(i)$, which an inductive argument shows can never exceed $\log_2N$, where $N$ is the system size. Newman-Ziff used a trick called {\em path compression} to make these paths---averaged over the execution of the algorithm---even smaller. When {\tt findroot(i)} traces its way to $r(i)$, noting that $r(i)$ is the root for each object $j$ along the path from $i$ to $r(i)$, it assigns $p(j) = r(i)$ for all those objects, linking each one directly to the root of the cluster. So next time we call {\tt findroot}, it would work in a single step. With this modification to {\tt findroot()}, the amortized cost of the {\tt findroot} and {\tt merge} operations (cost per operation, averaged over many operations), is {\em essentially} $O(N)$. To be precise, the amortized cost per step is proportional to the inverse of the Ackermann function $\alpha(N)$, which grows incredibly slowly with $N$~\cite{Tar75}.

\subsection{Simulation of layered percolation}

The first step in setting up the simulation for layered percolation is to create the nearest-neighbor matrix $A$, where $A(i, k) = j$ means node $j$ is the $k$-th neighbor of node $i$, $1 \le k \le d(i)$, where $d(i)$ is the degree of node $i$. As an illustration of the construction of the nearest-neighbor matrix, we show in Fig.~\ref{fig:kagome_nz} how we numbered the nodes for the kagome lattice, and assigned values to ${A(i, k)}$ for $0 \le i \le N-1$, $0 \le k \le 3$, for a kagome lattice with $N$ nodes each of degree, $d = 4$. 

Algorithm~\ref{alg:newmanziff} summarizes the remainder of the algorithm {\sc MultiLayerNewmanZiff}$(A,N,M)$, that takes as inputs, the nearest-neighbor matrix $A$, the number of nodes $N$ and the number of layers $M$, and produces estimates of the multilayer thresholds, $q_c(m)$, $m=1, \ldots, M$. The algorithm, as written below, accurately estimates the sizes of the largest cluster, ClusterSize$(m,i)$ for an $m$-layer merged lattice as a function of the single-layer site-occupation probability, $q  \equiv i/N$. For lattices, the maximum cluster sizes ClusterSize$(m,i)$ have a sharp discontinuity at $q = q_c(m)$ (i.e., $i = i_c \equiv Nq_c(m)$) for an infinite size lattice (see Fig.~\ref{fig:squaregrid_results_main}(a) for instance), for $N$ large enough. Therefore, the $i$ value where the discrete slope of the maximum cluster size (ClusterSize$(m,i)$ $-$ ClusterSize$(m,i-1)$) is maximum gives a pretty good estimate of $i_c(m)$ (hence, that of $q_c(m)$), which suffices for the purposes of this paper. A more accurate way to estimate the threshold $q_c(m)$ involves estimating the {\em wrapping probabilities}, which are probabilities that a cluster wraps around the lattice boundary conditions, either vertically, or horizontally, or both. 

\begin{algorithm}
\caption{\label{alg:newmanziff}{\sc MultiLayerNewmanZiff}$(A, N, M)$}
\begin{algorithmic}[1]
\Require  The nearest-neighbor matrix $A$ of a graph $G$, and the total number of layers, $M$. 

\For{$m=0$ to $M-1$}
	\State Initialize: cluster size dummy variable, $C(m)=0$;
	\State Generate a random permutation, $\pi_m()$ of $[1, \ldots, n]$;
	\For{$i=0$ to $N-1$}
		\State ptr$(m,i)$=EMPTY;
	\EndFor
\EndFor
\For{$i=0$ to $N-1$}
	\For{$m=0$ to $M-1$}	
		\State occupied$(m^\prime, \pi_{m^\prime}(i))=0$;
		\For{$m^\prime = 0$ to $m$}		
			\State $s_1=\pi_{m^\prime}(i)$;
			\State occupied$(m^\prime,s-1)=1$;

			\If {${\rm ptr}(m,s_1) \ne$ EMPTY}
				\State $r_1=$ {\sc FindRoot}$(s_1,m)$;
				\If {$-{\rm ptr}(m,r_1) >$ big$(m)$}
					\State big$(m)$= $-{\rm ptr}(m,r_1)$;
				\EndIf
			\Else
				\State ${\rm ptr}(m,s_1)=-1$;
				\State $r_1=s_1$;
			\EndIf
	
			\For{$j = 0$ to $d-1$}
				\State $s_2=A(s_1,j)$;
				\If {occupied$(m^\prime,s_2)=1$}
					\State $r_2=$ {\sc FindRoot}$(s_2,m)$;
					\If {$r_2 \ne r_1$}
						\If {${\rm ptr}(m,r_1)>{\rm ptr}(m,r_2)$}
							\State Increment ${\rm ptr}(m,r_2)$ by ${\rm ptr}(m,r_1)$;
							\State ${\rm ptr}(m,r_1) = r_2$;
							\State $r_1 = r_2$;
						\Else
							\State Increment ${\rm ptr}(m,r_1)$ by ${\rm ptr}(m,r_2)$;
							\State ${\rm ptr}(m,r_2) = r_1$;
						\EndIf
						\If {$-{\rm ptr}(m,r_1)>C(m)$}
							\State $C(m) = -{\rm ptr}(m,r_1)$;
						\EndIf
					\EndIf
				\EndIf
			\EndFor
		\EndFor
		\State ClusterSize$(m,i)=C(m)/N$;
	\EndFor
\EndFor		
\For {$m=0$ to $M-1$}
	\State slope$_{\rm max}$ $= 0$;
        \For{$i=1$ to $N-1$}
		\State slope = ClusterSize$(m,i)$-ClusterSize$(m,i-1)$;
		\If {slope $>$ slope$_{\rm max}$}
			\State slope$_{\rm max}$ = slope;
		        \State $q_c(m)=i/N$;
		\EndIf
	\EndFor
\EndFor
\State \Return multilayer percolation thresholds, $q_c(m)$, $0 \le m \le M-1$.

\Function{FindRoot}{$i, m$}
	 \If {${\rm ptr}(m,i)<0$} 
	 	\State \Return $i$;
    	\EndIf
\State \Return ${\rm ptr}(m,i) =$ {\sc FindRoot}$({\rm ptr}(m,i),m)$;
\EndFunction
\end{algorithmic}
\end{algorithm}

\end{document}